\documentclass[onecolumn]{IEEEtran}
\ifCLASSINFOpdf
\else
\fi

\usepackage{amsfonts}
\usepackage{amsmath,graphicx}
\usepackage{amssymb}
\usepackage{epsfig}
\usepackage{cite}
\usepackage{color, soul}
\usepackage{bm}
\usepackage{ifpdf}
\usepackage{mathrsfs,amsthm}
\usepackage{array}
\usepackage{multirow}
\usepackage{booktabs}
\usepackage{graphicx, subfigure}

\newtheorem{definition}{Definition}

\newtheorem{remark}{Remark}
\newtheorem{theorem}{Theorem}

\newtheorem{lemma}{Lemma} 
\newtheorem{corollary}{Corollary}
\newtheorem{design}{Design}

\begin{document}
%
\title{Complementary Lattice Arrays \\ for Coded Aperture Imaging}

\author{Jie~Ding,~\IEEEmembership{Student Member,~IEEE,}
	Mohammad~Noshad,~\IEEEmembership{Member,~IEEE,}
        and~Vahid~Tarokh,~\IEEEmembership{Fellow,~IEEE}
       \thanks{J. Ding, M. Noshad, and V. Tarokh are with the School of Engineering and Applied Sciences,
		Harvard University, Cambridge, MA 02138, USA. E-mail: jieding@g.harvard.edu , mnoshad@seas.harvard.edu , vahid@seas.harvard.edu .}
}

%

\maketitle

\begin{abstract}
In this work, we consider 
complementary lattice arrays
in order to enable a broader range of designs for coded aperture imaging systems.
We provide a general framework and methods that
generate richer and more flexible designs than existing ones.
Besides this, we review and interpret the state-of-the-art uniformly redundant arrays (URA) designs,
broaden the related concepts, and further propose some new design methods.
\end{abstract}


%
\IEEEpeerreviewmaketitle

\section{Introduction} \label{introSection}

\setcounter{equation}{0}
\renewcommand{\theequation}{\arabic{equation}}

\IEEEPARstart{I}{maging} using high-energy radiation with spectrum ranging from $X$-ray to $\gamma$-ray has found many applications including high energy astronomy
\cite{cook1984gamma,CAI_astronomy1987} and medical imaging \cite{Jap1978,Jap1981tomogram,chen2003tomographic}. In these wavelengths, imaging using lenses is not possible since the rays cannot be refracted or reflected, and hence cannot be focused. An alternative technique to do imaging in this spectrum is to use pinhole cameras, in which the lenses are replaced with a tiny pinhole. The problem in these cameras is that the pinholes passes a low intensity light, while for imaging purposes, a much stronger light is needed. Increasing the size of the pinhole cannot solve this problem as it increases the intensity at the expense of decreased resolution of the image. Coded aperture imaging (CAI) is introduced to address this issue by increasing the number of the pinholes. A coded aperture is a grating or grid that casts a coded image on a plane of detectors by blocking and unblocking the light in a known pattern, and produces a higher signal to noise ratio (SNR) of the image while maintaining a high angular resolution \cite{simpson1980coded,fenimoreCAI_URA1978}. The coded image is then correlated with a decoding array in order to reconstruct the original image. 
The deployment of pinholes 
and the decoding array are usually jointly designed such that perfect or near-perfect reconstruction is possible. 
Figure \ref{diagram} gives the schematic diagram of a CAI system.

\begin{figure}[h]
  \centering
  \includegraphics[width=0.5\textwidth]{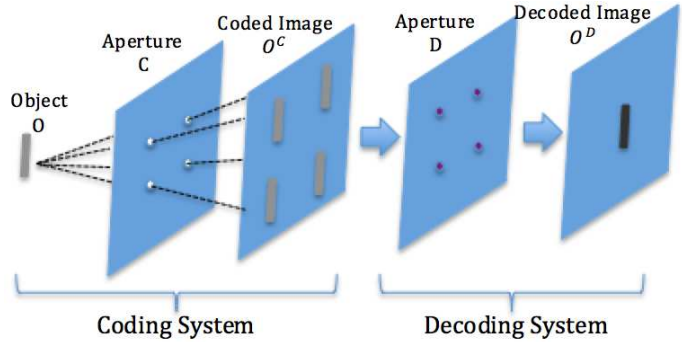}
  \caption{An illustration of the CAI system
  }
  \label{diagram}
\end{figure}

Since the coded aperture is usually defined based on an integer lattice, it can be modeled as a two dimensional array. 
In general, we define the $n$-dimensional array $C$ of complex-valued entries of size $L_1 \times \cdots \times L_n$:
\begin{align*}
C[c_1,\cdots, c_n]= 0, \quad \forall \,  c_1,\cdots, c_n \in \mathbb{Z}, (c_1,\cdots, c_n) \notin [0,L_1-1] \times \cdots [0,L_n-1].
\end{align*}
For simplicity, $C[c_1,\cdots, c_n]$ is also denoted by $C[\bm a]$, where $\bm a = [c_1,\cdots, c_n]^{\textrm{T}} \in \mathbb{R}^n$. 
The decoding array $D$ can be similarly defined. 
The set from which the elements of the aperture arrays take values from is referred to as an ``alphabet''.  In coded aperture imaging, a physically realizable coding aperture usually consists of $0/1$ binary alphabet (representing closed/open pinholes).
If multiple coded images are obtained with different aperture masks and the resulting digital projection images are suitably combined, a complex-valued array $C$ becomes applicable \cite{ohyama1983advanced,busboom1997coded}.  
For example, a $\{ \pm 1\}$ coding aperture could be obtained computationally from two masks with openings at $C$'s $1, -1$ locations.
For an $N$-phase alphabet, it is calculated that the number of $\{0,1\}$ masks needed is $(3N-1)/2$ for odd $N$, and $N $ for even N. The calculation follows from the observations that a root of unity in the form of $x+iy, xy \neq 0 \,(i^2=-1)$ requires two masks, a pair in the form of $(x+iy,x-iy), xy \neq 0$ require only three masks (a mask corresponding to $x$ is shared), a pair in the form of $(x+iy, -x+iy), xy \neq 0$ require only three masks (a mask corresponding to $iy$ is shared),  and $1,-1,i,-i$ each requires one mask.
Moreover, the development of hardware technology, e.g. spatial light modulators, 
may lead to realizable complex-valued physical masks. 
If both the coding and decoding systems use such masks, an analog reconstruction could be obtained.
Due to the above reasons, we assume that the elements of an aperture could be unimodular complex numbers.  

For a planar object that is projected onto a gamma camera through the coding aperture, it can be shown that the object is perfectly decoded if $C * D$ is a multiple $m$ of the discrete delta function $\delta[\bm r]$, 
where $*$ denotes the convolution and $\delta[\bm r]$ corresponds to an array with one centered at the origin and zero elsewhere \cite{fenimoreCAI_URA1978,Jap1978}.
The value of $m$ is also called ``the SNR gain''.
Clearly,  larger $m$ is better, and $m \leq \omega$, 
where $\omega  =L_1  \cdots  L_n$ (the number of pinholes).
Designing $C$ and $D$ is the key part of designing CAI.

It is useful to observe that the designs of apertures are intimately related to the concept of ``autocorrelation'',
and there are two typical types of ``autocorrelation'' for an array. 
Though the two types of autocorrelations are different, they can be both applied to the design of apertures through different approaches, as will be pointed out in this paper.
One is ``aperiodic autocorrelation''.
The aperiodic autocorrelation function $A^C(\cdot)$ is given by
\begin{align*}
A^C(v_1,\cdots, v_n) = \sum_{c_1,\cdots, c_n \in \mathbb{Z}} C[c_1,\cdots, c_n] \overline{C[c_1+v_1,\cdots, c_n+v_n]}, \quad v_1,\cdots, v_n \in \mathbb{Z},
\end{align*}
where $\bar{c}$ is the complex conjugate of $c$.
%
The other is ``periodic autocorrelation'', which is defined later in the paper. 
If we choose $D=C^{-}$, where $C^{-}[c_1,\cdots, c_n]=\overline{C[-c_1,\cdots, -c_n]}$, then $C*D $ gives the autocorrelation of $C$.

Non-redundant arrays (NRA) have been introduced for arranging the pinholes in CAI, since they have the property that the aperiodic autocorrelations consist of a central spike with the side-lobes equal to one within certain lag (range of the argument $v_1,\cdots,v_n$) and either zero or unity beyond the lag \cite{golay1971}. Pseudo-noise arrays (PNA) \cite{key2} are another alternative, whose periodic autocorrelations consist of a central spike with $-1$ side-lobes, which lead to designs of a pair of arrays such that their convolution is a multiple of the discrete delta function \cite{key1}. Twin primes, quadratic residues, and m-sequences are examples of PNA designs. NRA and PNA based designs are both referred to as uniformly redundant arrays (URA) \cite{fenimoreCAI_URA1978,key1,fenimore1978_performance}. However, the size of the URA structures is restricted and cannot be adapted to any particular detector \cite{baumert1971cyclic,CAI_astronomy1987}. Besides this, the SNR gain for URA is limited to $\omega/2$\cite{fenimoreCAI_URA1978,HURA1985hexagonal,MURA1989_linear_square_hex,byard1992optimised}.
Other designs that have also been used in CAI are geometric design \cite{gourlay1983geometric} and pseudo-noise product design \cite{gottesman1986pnp}, but they are also available only for a limited number of sizes---the former design is for square arrays, and the latter one requires that pseudo-noise sequences exist for each dimension.

Though it is generally hard to find a single pair of coding and decoding arrays, it may be easier to find several pairs which act perfectly while combined together.
Based on this idea, we look for a broader range of designs of the coding arrays in this paper.
We show that the aperture can then be customized to any shape on any lattice, meeting various demands in practical situations.
Our work is inspired by
Golay complementary arrays, which are defined as a pair of arrays whose aperiodic autocorrelations sum to zero in all out-of-phase positions. They have been used for pinhole arrangement in order to obtain the maximum achievable SNR gain, while eliminating the side lobes of the decoded image \cite{Jap1978}. 
We note that there is a natural mapping between a pair of Golay complementary arrays, say $C_1$ and $C_2$, and a CAI system consisting of two parallel coding/decoding apertures.
It is illustrated in Figure \ref{channel2}, where $D_1=C_1^{-},D_2=C_2^{-}$.
\begin{figure}[h]
  \centering
  \includegraphics[width=0.4\textwidth]{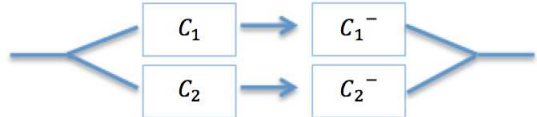}
  \caption{A CAI system with two parallel channels, with coding apertures
  $C_1$, $C_2$, and decoding apertures $C_1^{-}$, $C_2^{-}$}
  \label{channel2}
\end{figure}
When an object goes through the system,
the side lobes are completely canceled out by the addition of the two decoded images.

Though an aperture is usually defined based on an integer lattice, we consider the design problem in the context of a general lattice, naturally arising from practical implementations.
For example, usually the distance between two pinholes should be no less than a given threshold due to physical constraints.
 It has been shown by L. F. Toth \cite{circlePacking1940} that the lattice arrangement of circles
with the highest density in the two-dimensional Euclidean space is the hexagonal packing arrangement, in which the centers of the circles are arranged in a hexagonal lattice.
Thus, given the minimal distance allowed among pinholes, the most
compact arrangement  (thus with the largest possible SNR gain) is to arrange them on a hexagonal lattice.



The outline of this paper is as follows. 
In Section \ref{secTheory} we briefly present related work on Golay complementary arrays (based on aperiodic autocorrelation), and then
propose ``complementary lattice arrays'' and other related new concepts such as the ``complementary array banks'', including
Golay complementary array pairs as a special case.
This general framework naturally leads to the new concept of ``multi-channel CAI system'' which extends
the classical CAI system.
We provide the concept, theory, and the design framework.
Due to the reasons mentioned before, 
our examples are based on two-dimensional hexagonal arrays and unimodular alphabets (which consist of unimodular complex numbers).
Nevertheless, the methodology given in this work could be further generalized. 
In Section \ref{secURA} we review the URA literature that is mostly based on periodic autocorrelations.
We further generalize the related concepts 
in Section \ref{newPeriodic} in a similar fashion. 
This leads to a new class of aperture designs, which have the desirable imaging characteristics of URAs, yet exist for sizes for which URAs do not exist.
In Section \ref{secSim}, we provide computer simulations demonstrating the performance of our schemes.
\section{Concept, Theory and Design of Complementary Lattice Array}  \label{secTheory}

\subsection{Golay Complementary Arrays}  \label{sec2_Golay}

In this section we briefly introduce related works on Golay complementary array pairs.
Golay \cite{golayComplementarOrigin1951} first introduced Golay complementary sequence pairs in 1951 to address the optical problem of multislit spectrometry .
Later, they were also used for many other applications, including horizontal modulation systems in communication \cite{golayComplementary1961}, power control for multi-carrier wireless transmission \cite{davis1999peak}, Ising spin systems \cite{Golay_Spin_03}, channel-measurement \cite{Golay_Channel_06, Golay_Channel_01}, and orthogonal frequency division multiplexing \cite{Golay_OFDM_97}. 

A concept that is related to complementary array pairs is the Barker array.
It is a $\{\pm 1\}$ binary array $C$ such that
\begin{align}\label{collectiveSmall}
| A^C(v_1,\cdots, v_n) | \leq 1, \quad \forall \,  v_1,\cdots, v_n \in \mathbb{Z}, (v_1,\cdots, v_n) \neq (0, \cdots, 0).
\end{align}
Barker arrays are scarce.
For $n=1$, the known valid lengths are only $2$, $3$, $4$, $5$, $7$, $11$, and $13$.
%
For $n=2$, it has been proved that there is no Barker array for $L_1>1,L_2>1$ except
$ \left[
	\begin{array}{cc}
		1 & 1 \\
		1 & -1	
	\end{array}
  \right]
$. 
Another related concept is the NRA, which
 also satisfies the condition (\ref{collectiveSmall}). Its only difference compared with
Barker array is that it is $\{0,1\}$-binary.

Golay complementary array pairs address
the scarcity of Barker arrays and NRAs.
The basic idea of Golay complementary array pairs is to use the nonzero part of one autocorrelation to
``compensate'' the nonzero counterpart of the other \cite{golayComplementarOrigin1951}.
Specifically,
a pair of arrays $C_1$ and $C_2$ of size $L_1 \times \cdots \times L_n$ is a Golay complementary array pair, if the sum of their
aperiodic autocorrelations is a multiple of the discrete delta function, i.e.
\begin{align*}
A^{C_1}(v_1, \cdots, v_n) + A^{C_2}(v_1, \cdots, v_n) =0, \quad \, \forall \,  v_1, \cdots, v_n \in \mathbb{Z}, (v_1, \cdots, v_n) \neq (0, \cdots,0).
\end{align*}
%
The initial study of Golay complementary sequence pairs ($n=1$) was for the binary case.
Binary Golay complementary sequence pairs are known for lengths 2, 10 \cite{golayComplementary1961}, and 26 \cite{golay1962note}.
It has been shown that infinitely many lengths could be synthesized using the existing solutions
\cite{turyn1974hadamard}. Specifically, binary Golay complementary sequence pairs with length
 $2^{k_1}10^{k_2} 26^{k_3}$ exist, where $k_1,k_2,k_3$ are any nonnegative integers. Besides, no sequences of other lengths have been found.
Later on, larger alphabets were considered,  including
$2^n$-phase \cite{craigen2002complex}, 
$N$-phase for even $N$ \cite{paterson2000generalized}, 
the ternary case $\mathfrak{A}=\{-1,0,1\}$ \cite{gavish1994ternary,craigen2001theory,craigen2006further},
and the unimodular case \cite{budivsin1990new}.
Here, an alphabet $\mathfrak{A}$ is called $N$-phase, if it consists of  $N$th roots of unity, i.e. $\mathfrak{A}=\{\zeta: \zeta^N=1\}$; it is unimodular if $\mathfrak{A}=\{\zeta: |\zeta|=1\}$.

In 1978, Ohyama et al. \cite{Jap1978} constructed binary Golay complementary array pairs ($n=2$) of size $2^{k_1} \times 2^{k_2}$.
The size is then generalized to $2^{k_1}10^{k_2} 26^{k_3} \times 2^{k_4}10^{k_5} 26^{k_6}$,
where $k_j, j=1,\cdots, 6 $ are any nonnegative integers \cite{luke1985sets,dymond1992barker}. 

We look for broader concepts and designs than complementary array pairs.
The examples provided in this paper are for the two-dimensional case, 
but they can be easily generalized to higher dimensions.
We start with the definitions in the following section.

\subsection{Definitions and Notations}

\begin{definition} \label{latticeDef}
A \textbf{lattice} in  $\mathbb{R}^n$ is a subgroup of $\mathbb{R}^n$ which is generated from
a basis by forming all linear combinations with integer coefficients. In other words, a lattice $\mathfrak{L}$ in
$\mathbb{R}^n$ has the form
$$\mathfrak{L} = \left\{ \sum_{i=1}^n c_i \bm e_i \mid c_i \in \mathbb{Z} \right\},$$
where $\{\bm e_i\}_{i = 1}^n$ forms a basis of $\mathbb{R}^n$.
\end{definition}

For example, the integer lattice $\mathbb{Z}^2$ is generated from the basis
 $\bm e_1=(1,0),  \bm e_2=(0,1)$.
The hexagonal (honeycomb) lattice $\mathbb{A}_2$ is generated  from the basis
$\bm e_1=(1,0),\bm e_2=(-\frac{1}{2}, \frac{\sqrt{3}}{2})$.

A classical array is based on an integer lattice. We now give the definition of an array that
is based on a general lattice.

\begin{definition} \label{defineArray}
Let $\mathfrak{L}$ be a lattice.
A \textbf{lattice array} $C^{\mathfrak{L},\Omega,\mathfrak{A}}$ defined over alphabet $\mathfrak{A}$ and with support $\Omega$ is a mapping $C[\cdot]: \mathfrak{L} \rightarrow \mathfrak{A}$,
such that $C[\bm a] =0$ for all $\bm a \notin \Omega$ and $C[\bm a] \in \mathfrak{A}$ for all $\bm a \in \Omega$.
The number of the elements of
$\Omega$ (array size) is  denoted by $|\Omega|$. 
We denote $C^{\mathfrak{L},\Omega,\mathfrak{A}}$ by $C$ when there is no ambiguity.
In other words,
\begin{align*}
C[\bm a]= 0,  \quad \forall \,  \bm a \in \mathfrak{L} \setminus \Omega,
\end{align*}
where $C[\bm a]$ is the entry at location $\bm a$.

The following terms are made to simplify the notations.
\begin{itemize}
\item 
Define $C^{\mathfrak{L},\Omega+\bm t,\mathfrak{A}} \{\bm t\}$ as the \textbf{shifted copy} of
$C^{\mathfrak{L},\Omega,\mathfrak{A}} $ by $\bm t$ (for $\bm t \in \mathfrak{L}$), if
\begin{align*}
C^{\mathfrak{L},\Omega+\bm t,\mathfrak{A}} \{\bm t\} [\bm a] = C^{\mathfrak{L},\Omega,\mathfrak{A}}[\bm a-\bm t],
\forall \,  \bm a \in \mathfrak{L}.
\end{align*}
For brevity, $C^{\mathfrak{L},\Omega+\bm t,\mathfrak{A}} \{\bm t\} $ is simplified as $C\{\bm t\}$.

\item 
Assume that two arrays $C_1^{\mathfrak{L},\Omega_1,\mathfrak{A}_1}$ and $C_2^{\mathfrak{L},\Omega_2,\mathfrak{A}_2}$
are based on the same lattice $\mathfrak{L}$, but not necessarily on the same area.
The \textbf{addition} of $C_1$ and $C_2$, $C = C_1 + C_2$, is an array whose entries are the
addition of corresponding entries in $C_1$ and $C_2$, i.e.
\begin{align*}
\Omega &= \Omega_1 \cup \Omega_2,   \quad
C[\bm a] = C_1[\bm a] + C_2[\bm a], \quad \forall \,  \bm a \in \Omega. 
\end{align*}

\item 
A set of arrays $\{ C_m^{\mathfrak{L},\Omega_m,\mathfrak{A}_m} \}_{m=1}^M$ are \textbf{non-overlapping} if
$$
\Omega_{m_1} \cap \Omega_{m_2} = \emptyset, \quad \forall \,  m_1,m_2 \in {1,2,\cdots, M}, m_1 \neq m_2.
$$


\end{itemize}

\end{definition}

\begin{definition} \label{defineAutocorrelation}
Assume that the lattice $\mathfrak{L}$ is generated from $\{\bm e_i\}_{i = 1}^n$.
The \textbf{aperiodic autocorrelation} function is
\begin{align*}
A^C(v_1,\cdots, v_n) =\sum_{\bm a \in \Omega} C[\bm a]  \overline{ C[\bm a + v_1 \bm e_1 + \cdots + v_n \bm e_n]},
\quad   v_1,\cdots, v_n \in \mathbb{Z}.
\end{align*}
The aperiodic crosscorrelation function $A^{C_1 C_2} (\cdot)$ of two arrays $C_1$ and $C_2$ is 
\begin{align*}
A^{C_1 C_2}(v_1,\cdots, v_n) =\sum_{\bm a \in \Omega} C_1[\bm a] \overline{C_2[\bm a + v_1 \bm e_1 + \cdots + v_n \bm e_n]}, \quad v_1,\cdots, v_n \in \mathbb{Z}.
\end{align*}
Sometimes, $A^C(\cdot)$ and $A^{C_1 C_2}(\cdot)$ are  respectively denoted by $C*C^{-}$ and $C_1 * C_2^{-}$.
\end{definition}

\begin{definition} \label{bankDef}
A \textbf{complementary array bank} consists of pairs
$\{(C_m^{\mathfrak{L},\Omega_{1m},\mathfrak{A}}, D_m^{\mathfrak{L},\Omega_{2m},\mathfrak{A}})\}_{m=1}^M$
 such that the sum of the crosscorrelations is a multiple of the discrete delta function:
$$
\sum_{m=1}^M C_m * D_m^{-} = \sum_{m=1}^M A^{C_m  D_m}(\cdot) = \omega \delta[\bm r],
$$
where $\omega$ is a constant.
$M$ is called the order, or number of channels.
\end{definition}

\begin{remark}
There is a natural mapping between a complementary array bank, say  $\{(C_m, D_m)\}_{m=1}^M$,
and a CAI system consisting of $M$ parallel channels, each of which consists of
a pair of coding and decoding apertures (Fig.  \ref{bankDiagram}).
When a source image comes, it is coded and decoded through $M$ channels simultaneously, and is then
retrieved by simply adding the decoded images from all the channels.
The multi-channel CAI system proposed here provides a generalized solution to CAI design,
by including a classical CAI  system as a special case.
In the remaining part of Section \ref{secTheory}, we mainly study the complementary array sets, which
may provide insight into the theory and design of complementary array banks in general.
\end{remark}
\begin{figure}[h]
  \centering
  \includegraphics[width=0.4\textwidth]{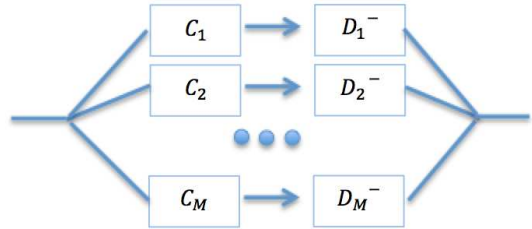}
  \caption{A CAI system with $M$ parallel channels, with coding apertures
  $C_1, \cdots, C_M$, and decoding apertures $D_1^{-}, \cdots, D_M^{-}$
  }
  \label{bankDiagram}
\end{figure}

\begin{definition} \label{defGeneralCA}
A set of arrays $\{C_m^{\mathfrak{L},\Omega_m,\mathfrak{A}}\}_{m=1}^M$ is a
\textbf{complementary array set}, if the sum of their aperiodic
autocorrelations is a multiple of the discrete delta function, i.e.
\begin{align}
&\sum_{m=1}^M A^{C_m}(v_1,\cdots, v_n) =0, \quad \forall \, v_1,v_2 \in \mathbb{Z}, (v_1,\cdots, v_n) \neq (0,\cdots,0). \label{eqn11}
\end{align}
A  Golay complementary array pair is the special case when $M=2$.
\end{definition}

\begin{remark}
In practice, the pinholes on an aperture may only change the phase of a source point.
Therefore, we assume a unimodular alphabet by default.
It is clear that if a set of non-overlapping arrays are based on unimodular/$N$-phase alphabets,
the addition of them is also based on an unimodular/$N$-phase alphabet.

The autocorrelation of any array is the same as that of its shifted copy.
This is because
for any  $v_1,\cdots, v_n \in \mathbb{Z}, \bm t  \in \mathfrak{L}$, we have
\begin{align*}
A^{C\{\bm t\}}(v_1,\cdots, v_n)
&=\sum_{\bm a \in \Omega + \bm t} C\{ \bm t\}[\bm a]  \overline{ C\{ \bm t \} [\bm a + v_1 \bm e_1 + \cdots + v_n \bm e_n]}
=\sum_{\bm a \in \Omega + \bm t} C[\bm a - \bm t]  \overline{ C[\bm a -\bm t + v_1 \bm e_1 + \cdots + v_n \bm e_n]}  \\
&=\sum_{\bm a \in \Omega } C[\bm a]  \overline{ C[\bm a + v_1 \bm e_1 + \cdots + v_n \bm e_n]}
= A^{C}(v_1,\cdots, v_n).
\end{align*}
%
Furthermore, if $\{C_m\}_{m=1}^M$ is a complementary array set, then $\left\{C_m\{\bm t_m\} \right\}_{m=1}^M, \,\forall \,
\bm t_m \in \mathfrak{L}$ also forms a complementary array set. In other words, a complementary array set is ``invariant''
under shift operation.

Based on a unimodular alphabet, a complementary array set
$\{C_m\}_{m=1}^M$ satisfies
$$
\sum_{m=1}^M A^{C_m}(0,\cdots, 0) = M|\Omega|.
$$
Thus, the sum of the autocorrelations may be written as a multiple of the discrete delta function:
\begin{align*}
\sum_{m=1}^M C_m * C_m^{-} = \sum_{m=1}^M A^{C_m}(\cdot) = M |\Omega| \delta[\bm r].
\end{align*}

\end{remark}

\subsection{Inspirations from Ohyama et al.'s Design}

In Ohyama et al.'s design, $\mathfrak{L}$ is an integer lattice, and the number of complementary arrays is $M=2$.
The design consists of two steps:

\begin{itemize}
\item First, choose the following complementary sequence pair:
         \begin{align} \label{squareSeed}
         C_1 = [ 1 \quad 1], \quad C_2 = [1 \quad -1].
         \end{align}

\item Second, design complementary array pairs of larger sizes in an inductive manner.
         Assume that we already have a complementary pair $C_1, C_2$,
         with $C_1 * C_1^{-} + C_2 * C_2^{-} = 2\omega \delta[\bm r]$,
         where $\omega$ is constant.
          Let
         \begin{align}
         \widehat{C}_1 &= C_1\{\bm t_1\} + C_2\{\bm t_2\}, \quad 
         \widehat{C}_2 = C_1\{\bm t_1\} - C_2\{ \bm t_2\},\label{squareInduction2}
         \end{align}
where the shifts $\bm t_1$ and $\bm t_2$ are arbitrarily chosen.
\end{itemize}

The validity of the construction 
(\ref{squareInduction2}) is clear
    from the fact that
         \begin{align*} 
       \widehat{C}_1 * \widehat{C}_1^{-}
      = &C_1\{\bm t_1\} * C_1\{\bm t_1\}^{-} + C_2\{\bm t_2\} * C_2\{\bm t_2\}^{-}
          + C_1\{\bm t_1\} * C_2\{\bm t_2\} ^{-}+ C_2\{\bm t_2\} * C_1\{\bm t_1\}^{-},  \\
        \widehat{C}_2 * \widehat{C}_2^{-}
        =& C_1\{\bm t_1\} * C_1\{\bm t_1\}^{-} + C_2\{\bm t_2\} * C_2\{\bm t_2\}^{-}
          - C_1\{\bm t_1\} * C_2\{\bm t_2\}^{-} - C_2\{\bm t_2\} * C_1\{\bm t_1\}^{-},
         \end{align*}
         and thus
         \begin{align}
         & \widehat{C}_1 * \widehat{C}_1^{-} + \widehat{C}_2 * \widehat{C}_2^{-}
          = 2 (C_1\{\bm t_1\} * C_1\{\bm t_1\}^{-} + C_2\{\bm t_2\} * C_2\{\bm t_2\}^{-} )
          = 2(C_1 * C_1^{-} + C_2 * C_2^{-}) = 4 \omega \delta[\bm r].\label{squareInduction3}
         \end{align}
In practice, $\bm t_1$ and $\bm t_2$ are chosen properly so that $C_1\{\bm t_1\}$ and $C_2\{\bm t_2\}$ do not overlap,
which guarantees that $\widehat{C}_1$ and $\widehat{C}_2$  are still based on unimodular alphabets.
For example, after applying Equation (\ref{squareInduction2}) to (\ref{squareSeed}) once, we have two possible complementary array pairs:
\begin{align}
\widehat{C}_1 = [ 1 \quad 1 \quad 1 \quad -1], \quad \widehat{C}_2 = [1 \quad 1 \quad -1 \quad 1];
\end{align}
or
\begin{align}
 \widehat{C}_1= \left[
	\begin{array}{cc}
		1 & 1 \\
		1 & -1	
	\end{array}
  \right],  \quad
 \widehat{C}_2 = \left[
	\begin{array}{cc}
		1 & 1 \\
		-1 & 1	
	\end{array}
  \right].
\end{align}
The process of design is also illustrated in Fig. \ref{squareGrowDemo}.
\begin{figure}[h]
  \centering
  \includegraphics[width=0.6\textwidth]{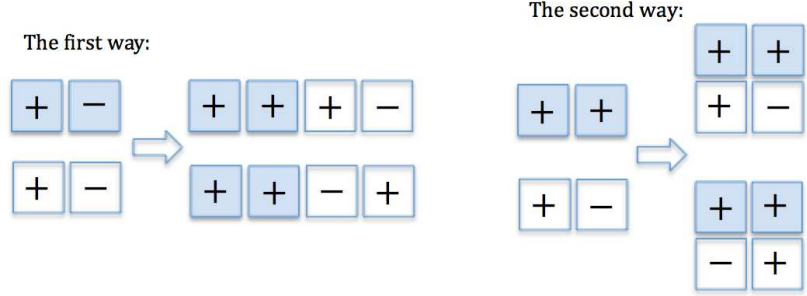}
  \caption{Illustration of the design of complementary array pairs on integer lattices}
  \label{squareGrowDemo}
\end{figure}
By continuous application of the above design process, complementary pairs of size $2^{k_1} \times 2^{k_2}$
(for any nonnegative integers $k_1$, $k_2$) can be designed.

Inspired by the above design for complementary array pairs on integer lattices, we look for a ``seed'' (similar to (\ref{squareSeed})) and a related scheme to ``grow'' the seed (similar to (\ref{squareInduction2})) 
for the design of complementary hexagonal arrays.
Admittedly, we may build a simple mapping between two-dimensional arrays on a square lattice and a hexagonal lattice (or other lattices) below:
\begin{align}
C^{\mathfrak{L}^s,\Omega^s,\mathfrak{A}} [c_1 * \bm e_1^s + c_2 * \bm e_2^s]
= C^{\mathfrak{L}^h,\Omega^h,\mathfrak{A}} [c_1 * \bm e_1^h + c_2 * \bm e_2^h], \quad \forall \,  c_1, c_2 \in \mathbb{Z},
\end{align}
where the superscripts $s$ and $h$ respectively denote square and hexagonal lattices.
Under the above mapping, a set of complementary arrays on a square lattice are still complementary on a hexagonal lattice.
This is due to the fact that the autocorrelation of an array is only with respect to the coefficients $c_1,c_2$.
Nevertheless, the lattice array naturally arises from practical designs.
Consider the scenario where a two (or three)-dimensional coded aperture is to be built that has pinholes arranged on a certain (suitably chosen) type of lattices which adapts to a particular physical aperture mask. The designs may preferably be based directly on that lattice instead of mapping to a square lattice (with zero elements padded in various regions) first and then mapping back to the original one.


\subsection{Design for The Basic Hexagonal Array of $7$ Points} \label{7pointSection}

We first study a very simple hexagonal pattern that may act as a ``seed''. It is an hexagonal array of $7$ points,
which is shown in Fig. \ref{7point}. After that, we consider possible ways to ``grow'' the seed.

\begin{figure}[h]
  \centering
  \includegraphics[width=0.15\textwidth]{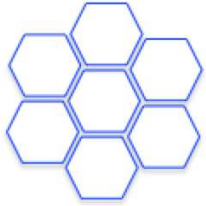}
  \caption{Basic hexagonal array with 7 points}
  \label{7point}
\end{figure}
We start from considering the existence of hexagonal complementary array pairs, i.e.
the order $M$ is $2$.

\begin{theorem} \label{M2NonexistenceThm}
For the basic 7-points hexagonal array, there exists no complementary array pair with unimodular alphabet
(Fig. \ref{pair_design})
\end{theorem}

\begin{figure}[h]
  \centering
  \includegraphics[width=0.3\textwidth]{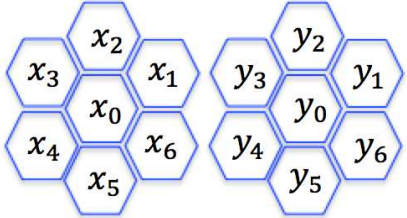}
  \caption{There exists no such complementary pair with unimodular alphabet.}
  \label{pair_design}
\end{figure}


The proof is given in Appendix \ref{appendix_a}.
One may be further interested in the existence of a hexagonal complementary array pair if the basic array does not have the origin $0$
(Fig. \ref{6Points}). 
In fact, it does not exist, either.

\begin{theorem} \label{M26PointNonexistenceThm}
For the array in Fig. \ref{6Points}, there exists no hexagonal complementary array pair with unimodular alphabet.
\end{theorem}


\begin{figure}[h]
  \centering
  \includegraphics[width=0.3\textwidth]{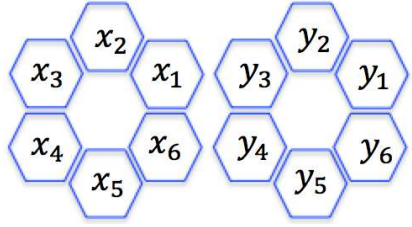}
  \caption{Basic hexagonal array with 6 points}
  \label{6Points}
\end{figure}
The proof is given in Appendix \ref{appendix_a2}.
The non-existence of complementary array pairs for the array
in Fig. \ref{7point} motivates us to
further consider higher 
order $M$.
We use the notation of ``design parameter'' for brevity.
For a particular array pattern, if there is a complementary array set with $M$ arrays and an $N$-phase
alphabet, the pair $(M,N)$ is called its \textbf{design parameters}. Furthermore, if the array sizes are equal to $L$,
we may refer to the triplet $(M,N,L)$ as its \textbf{design parameters} whenever there is no ambiguity.

Fortunately, complementary array triplets with unimodular alphabet exist.
In fact, we have found more than one designs with $(M,N,L) = (3,3,7)$. The following is an example.

\begin{design} \label{3_3design}
Let $\zeta = \exp ( i 2 \pi / 3)$. Let $C_1=\{ x_k \}_{k=0}^6, C_2= \{ y_k \}_{k=0}^6, C_3= \{ z_k \}_{k=0}^6$
denote the entries of three hexagonal arrays shown in Fig. \ref{7point}.
Then
\begin{align*} 
\{ x_k \}_{k=0}^6 &= \{ \zeta^2, \zeta^0, \zeta^2, \zeta^2, \zeta^0, \zeta^2, \zeta^0 \}, \quad
\{ y_k \}_{k=0}^6 = \{ \zeta^1, \zeta^0, \zeta^2, \zeta^2, \zeta^1, \zeta^0, \zeta^1 \}, \quad
\{ z_k \}_{k=0}^6 = \{ \zeta^1, \zeta^0, \zeta^1, \zeta^1, \zeta^2, \zeta^0, \zeta^1 \}
\end{align*}
form a complementary array set (Fig. \ref{3_3designFig}).
\end{design}

\begin{figure}[h]
  \centering
  \includegraphics[width=0.45\textwidth]{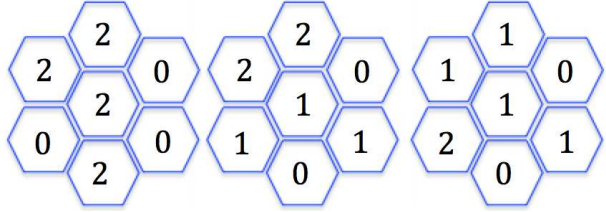}
  \caption{A complementary triplet with $3$-phase alphabet, i.e. $(M,N,L)=(3,3,7)$,
   where $\zeta^k$ is represented by $k$ , for $k=0,1,2$}
  \label{3_3designFig}
\end{figure}

We have also found more than one designs with $(M,N,L) = (4,2,7)$. The following is an example.

\begin{design} \label{4_2design}
Let $C_1=\{ x_k \}_{k=0}^6, C_2= \{ y_k \}_{k=0}^6, C_3= \{ z_k \}_{k=0}^6, C_4= \{ w_k \}_{k=0}^6$
denote the entries of four hexagonal arrays shown in Fig. \ref{7point}.
Then
\begin{align*}
\{ x_k \}_{k=0}^6 &= \{ 1, 1, -1, 1, 1, -1, 1 \}, \quad
\{ y_k \}_{k=0}^6 = \{ -1, 1, 1, -1, -1, 1, 1 \}, \\
\{ z_k \}_{k=0}^6 & = \{ 1, 1, -1, 1, -1, 1, 1 \},  \quad
\{ w_k \}_{k=0}^6 = \{ 1, 1, 1, -1, 1, -1, 1 \}
\end{align*}

form a complementary array set (Fig. \ref{4_2designFig}).
\end{design}

\begin{figure}[h]
  \centering
  \includegraphics[width=0.3\textwidth]{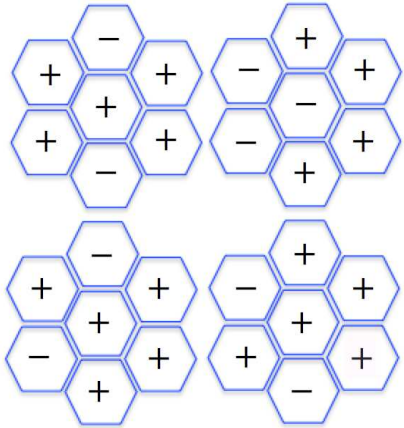}
  \caption{A complementary quadruplet with $2$-phase (binary) alphabet, i.e. $(M,N,L)=(4,2,7)$,
  where $\pm 1$ is represented by $\pm$ }
  \label{4_2designFig}
\end{figure}

\subsection{Methodology for Designing Larger Arrays} \label{generalMethod}

We now consider how to ``grow'' the seed that we have found in order to design more and larger arrays.
 If $+$ and $*$ are respectively considered as addition and multiplication operations, then (\ref{squareInduction2}) and (\ref{squareInduction3}) can be written in symbolic expression:
 \,
 \begin{align} 
   \bm{\widehat{C}} &= \bm H \bm C\{\bm t\}, \\
   \bm{\widehat{C}}^T \bm{\widehat{C}}^{-} &=  \bm C\{\bm t\}^T \bm H^T \bm H \bm C\{\bm t\}^{-} = 2 \bm C\{\bm t\}^T \bm C\{\bm t\}^{-},
 \end{align}
 where
 \begin{align*}
 &\bm{\widehat{C}} =
    \left[
	\begin{array}{cc}
		\widehat{C}_1 \\
		\widehat{C}_2 	
	\end{array}
  \right],  \quad
    \bm{\widehat{C}}^{-} =
    \left[
	\begin{array}{cc}
		\widehat{C}_1^{-} \\
		\widehat{C}_2^{-} 	
	\end{array}
  \right],  \quad
  \bm H =
  \left[
	\begin{array}{cc}
		1 & 1 \\
		1 & -1	
	\end{array}
  \right],  \quad
  \bm C\{\bm t\} =
      \left[
	\begin{array}{cc}
		C_1\{\bm t_1\} \\
		C_2\{\bm t_2\} 	
	\end{array}
  \right],
 \end{align*}
and $\bm H^T$ is the  conjugate transpose of $\bm H$.
 The key that $\bm{\widehat{C}}$ remains to be a complementary pair is that $\bm H$ satisfies
 $\bm H^T \bm H = 2 \bm I$, i.e. $\bm H$ is a Hadamard matrix.
This observation could be generalized to the following result.

\begin{theorem} \label{growScheme}
Let $\bm U = [u_{mk}]_{M \times M}$ be a unitary matrix up to a constant,
i.e. $\bm U^T \bm U = c  \bm I$, where $c>0$.
Assume $\{ C_k^{\mathfrak{L},\Omega_k,\mathfrak{A}} \}_{k=1}^M$ is a complementary array set.
 Then,
$\{ \widehat{C}_m \}_{m=1}^M$  is also a complementary array set, where
\begin{align*}
 \widehat{C}_m =\sum_{k=1}^M u_{mk}  \cdot C_k^{\mathfrak{L},\Omega_k+\bm t_k,\mathfrak{A} } \{\bm t_k\},\, m=1,2,\cdots,M,
\end{align*}
 $ \bm t_1, \bm t_2, \cdots  \bm t_M$ are arbitrarily chosen,
and $u \cdot  C$ is an array that multiplies each entry of $C$ by the scalar $u$.
\end{theorem}

\begin{proof}
Define a vector space on $\mathfrak{A}$\footnote{
It is not necessarily a field.
} with the addition operation $+$ defined
in Definition \ref{defineArray}, and the variables
\begin{align*}
 \bm{\widehat{C}} =
    \left[
	\begin{array}{cc}
		\widehat{C}_1 \\
		\vdots \\
		\widehat{C}_M 	
	\end{array}
  \right],
  \bm C\{\bm t\} =
      \left[
	\begin{array}{cc}
		C_1\{\bm t_1\} \\
		\vdots \\
		C_M\{\bm t_M\} 	
	\end{array}
  \right].
 \end{align*}
 Define a quadratic form with the multiplication operation $*$ defined as the convolution.
The sum of aperiodic autocorrelations of $\{ \widehat{C}_m \}_{m=1}^M$ is
\begin{align}
 \sum_{m=1}^M  \widehat{C}_m * \widehat{C}_m^{-}
 &= \bm{\widehat{C}}^T \bm{\widehat{C}}^{-} = \bm C\{\bm t\}^T \bm U^T \bm U \bm C\{\bm t\}^{-}
 = c \bm C\{\bm t\}^T \bm C\{\bm t\}^{-} \nonumber \\
 &= c  \sum_{m=1}^M  C_m\{\bm t_M\} * C_m\{\bm t_M\}^{-}
 =c  \sum_{m=1}^M C_m * C_m^{-} = c M \omega \delta[\bm r]. \label{new_add50}
 \end{align}
\end{proof}
\begin{remark} \label{remark_chooseF}
We may be more interested in the case where
\begin{enumerate}
 \item $\mathfrak{A}$ is an $N$-phase alphabet;

 \item $\bm U$ is the Fourier matrix:
 $\bm F_M = [f_{mk}]_{M \times M}, \, f_{mk} = \exp \{ i 2 \pi (m-1)(k-1)/ M \}$;

 \item The shifted arrays $C_k^{\mathfrak{L},\Omega_k + \bm t_k,\mathfrak{A}}\{\bm t_k\}, k=1,\cdots, M$ do not overlap. The complementary array set  $\{C_k\}_{k=1}^{M}$ with  design parameters $(M,N,|\Omega|)$ (if $|\Omega_k|=|\Omega|, k=1,\cdots,M$) becomes $\{\widehat{C}_m\}_{m=1}^{M}$ with design parameters $(M, {\rm lcm}(N,M), M |\Omega|)$ according to Theorem 3,
 where  ${\rm lcm}$ stands for least common multiple.
 Besides this, we have $c=M$, $\omega = |\Omega|$ in (\ref{new_add50}).
\end{enumerate}

\end{remark}

Theorem \ref{growScheme} provides a powerful tool to design more complex complementary arrays, which will be
illustrated in Subsection \ref{designMoreSection}. For future reference, we also include the following fact.

\begin{remark} \label{catDesign}
Assume that $\bm C^{(1)}, \cdots, \bm C^{(K)} $
are $K$ complementary array sets on the same lattice,
and the set $\bm C^{(k)}$ has design parameters $(M_k, N_k)$,
$\forall \,  k=1,\cdots,K$.
Then, it is clear that $\bm C = \bm C^{(1)} \cup \cdots \cup \bm C^{(K)} $ can be thought as a new complementary set with
design parameters $(\sum_{k=1}^K M_k, {\rm lcm}(N_1,\cdots,N_K))$.
\end{remark}

\subsection{An Example---Design for the hexagonal array of 18 points} \label{designMoreSection}


The basic hexagonal array of $7$ points studied in Subsection \ref{7pointSection} contains two layers,
with $1$ and $6$ points, respectively. We now study the hexagonal array that contains one more layer
(shown in Fig. \ref{19points}).
The design for this hexagonal array is not very obvious, so we delete the center element
\footnote{In fact, designs always exist for an arbitrarily shaped array, as will be discussed later. We delete
the center point primarily for a cute solution.}, i.e.
layer~1. The remaining $18$ points are grouped into six basic triangular arrays:
$C_1, C_2, C_3, C_1^{'}, C_2^{'}, C_3^{'}$,
which is shown in Fig. \ref{18points}.
\begin{figure}[h]
  \centering
  \includegraphics[width=0.35\textwidth]{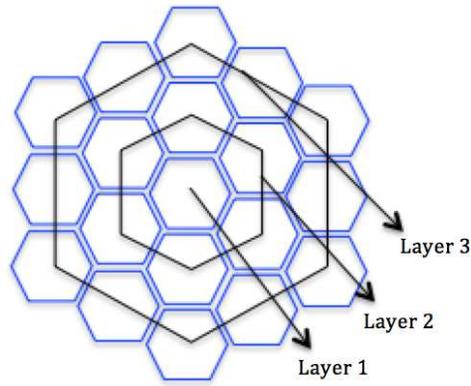}
  \caption{Hexagonal array with three layers (19 points)}
  \label{19points}
\end{figure}
\begin{figure}[h]
  \centering
  \includegraphics[width=0.29\textwidth]{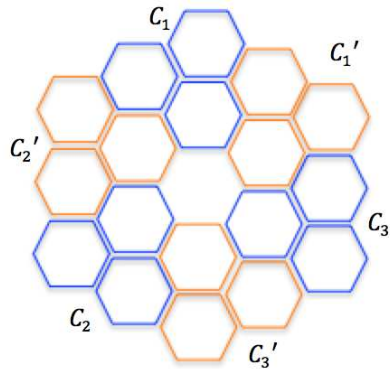}
  \caption{Hexagonal array with layer 2 and 3 (18 points)}
  \label{18points}
\end{figure}

This motivates the design of complementary array triplets for the basic triangular array shown in Fig. \ref{triangleBasic},
where $\zeta^k$ is represented by $k$ , for $k=0,1,2, \zeta = \exp ( i 2 \pi / 3)$.
In Fig. \ref{triangleBasic}, $\{C_1, C_2, C_3\}$ form a complementary array set.
Due to symmetry, its rotated copy,
$C_1^{'}, C_2^{'}, C_3^{'}$, also form a  complementary array set.
Applying Theorem \ref{growScheme}, we obtain the following design.
\begin{figure}[h]
  \centering
  \includegraphics[width=0.35\textwidth]{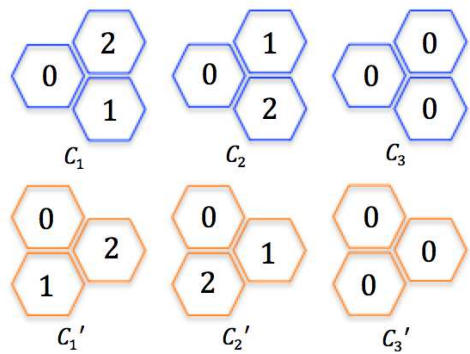}
  \caption{Basic triangular complementary array triplets with $3$-phase alphabet, i.e. $(M,N,L)=(3,3,3)$}
  \label{triangleBasic}
\end{figure}

\begin{design}

Let
\begin{align*}
 \bm{\widehat{C}} =
    \left[
	\begin{array}{cc}
		\widehat{C}_1 \\
		\widehat{C}_2 \\
		\widehat{C}_3 \\
		\widehat{C}_4 \\
		\widehat{C}_5 \\
		\widehat{C}_6 	
	\end{array}
  \right]
  = \bm F_6 \,
      \left[
	\begin{array}{cc}
		C_1^{'}\{\bm t_1\} \\
		C_1\{\bm t_2\} \\
		C_2^{'}\{\bm t_3\} \\
		C_2\{\bm t_4\} \\
		C_3^{'}\{\bm t_4\} \\
		C_3\{\bm t_4\}
	\end{array}
  \right],
 \end{align*}
 where $\bm t_1, \cdots, \bm t_6$ are such that $\widehat{C}_1, \cdots, \widehat{C}_6$ are arranged to form a hexagonal array of $18$ points,
 as is shown in Fig.~\ref{18points}. Then $\{\widehat{C}_m\}_{m=1}^6$ forms a complementary  array set. The design is also
 shown in  Fig.~\ref{18pointDesign}, where $\zeta^k$ is represented by $k$, for $k=0,\cdots, 5,
 \zeta = \exp ( i 2 \pi / 6)$.
\end{design}

\begin{figure*}
  \centering
  \includegraphics[width=0.75 \textwidth]{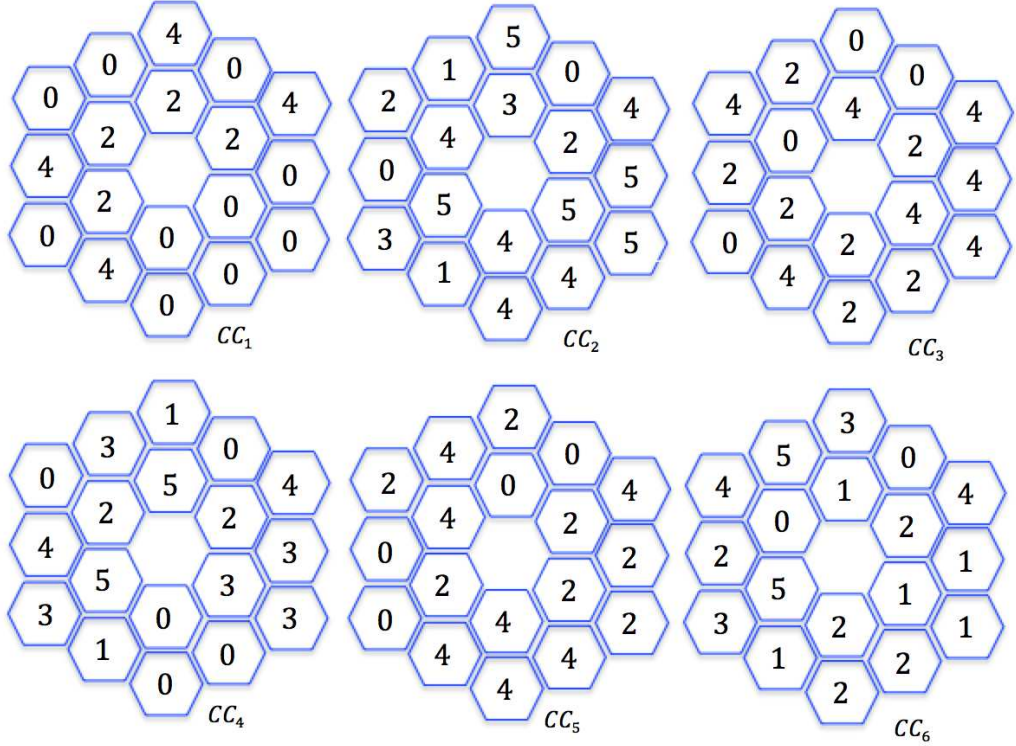}
  \caption{Complementary array set with parameter $(M,N,L)=(6,6,18)$}
  \label{18pointDesign}
\end{figure*}
\subsection{Complementary Array Bank} \label{bank}

Up to this point, we have assumed that the coding array $C$ and decoding array $D$ are related via $D[\bm r] = \overline{C[-\bm r]}$.
The design of CAI thus reduces to the design of complementary array sets.
Then we extend the autocorrelation to crosscorrelation, and
the design of complementary array sets is accordingly extended to that of complementary array banks.
The following result is a generalization of Theorem \ref{growScheme}, and its proof is similar to that of Theorem \ref{growScheme}.

\begin{theorem} \label{growSchemeGeneralization}

Let $\bm \Theta = [\theta_{mk}] , \bm \Phi = [\phi_{mk}] \in \mathbb{C}^{M \times \tilde{M}}$ be two matrices satisfying
$\bm \Theta^{T} \bm \Phi = c \bm I$ for some positive constant $c$.
For a given lattice $\mathfrak{L}$,
assume
   $\{(C_k, D_k) \}_{k=1}^M$
 is a complementary array bank. Then,
   $\{(\widehat{C}_m, \widehat{D}_m) \}_{m=1}^{\tilde{M}}$
is also a complementary array bank, where
\begin{align*}
&\widehat{C}_m =\sum_{k=1}^M \theta_{mk}  \cdot C_k\{\bm t_k\}, \quad
\widehat{D}_m =\sum_{k=1}^M \phi_{mk}  \cdot D_k\{\bm t_k\}, \quad
 m=1,2,\cdots,\tilde{M},
\end{align*}
 $ \bm t_1, \bm t_2, \cdots  \bm t_M$ are arbitrarily chosen,
and $u \cdot  C$ is an array that multiplies each entry of $C$ by the scalar $u$.
\end{theorem}

\begin{remark}  \label{remark_chooseU}
We may be more interested in the case where
\begin{enumerate}
 \item $\{C_k, D_k\}_{k=1}^{M}$ have $N$-phase alphabets;

 \item $\bm \Theta $ is equal to $\bm \Phi$  and it contains $M$ orthogonal columns of the complex Fourier matrix $F_{\tilde{M}}$ (thus $M \leq \tilde{M}$);

 \item The shifted arrays $C_k\{\bm t_k\}$ do not overlap, neither do $D_k\{\bm t_k\}, k=1,\cdots,M$.

\end{enumerate}

Furthermore, if we assume that $C_k = D_k, k=1, \cdots, M$ and the array sizes are equal to $L$,  then the complementary array set $\{C_k\}_{k=1}^{M}$ has design parameters $(M,N,L)$, and $\{\widehat{C}_m\}_{m=1}^{M}$ has design parameters $(\tilde{M}, {\rm lcm}(N,\tilde{M}), ML)$. 

\end{remark}

\begin{lemma} \label{singleLemma}
Any single point, as the simplest array on any lattice, forms a complementary set.
\end{lemma}

\begin{remark}
Lemma \ref{singleLemma} is a trivial but useful result.
It follows from Definition \ref{defGeneralCA}  and
 the fact that the aperiodic autocorrelation of a single point is always the discrete delta function.
By employing Lemma \ref{singleLemma}, Theorem \ref{growScheme}, and
Theorem \ref{growSchemeGeneralization},
it is possible to design complementary arrays of various support $\Omega$.

\end{remark}

\begin{corollary}
For an arbitrary set $\Omega$ on a lattice,  there exists at least one complementary array set with support $\Omega$
for any 
order $M$ such that $M \geq |\Omega|$.
\end{corollary}

\begin{remark} \label{augmentationRemark} [Augmentation using Theorem \ref{growSchemeGeneralization}]
Due to Theorem \ref{growSchemeGeneralization}, we may let $\tilde{M} > M$ for practical purposes.
For example, we may choose $\tilde{M} = 2^m$ ( for a positive integer $m$) such that that $\bm U $ is a $\pm 1$
Hadamard matrix and the growth of alphabet ($N$) could be well controlled.
We call this ``augmentation'' procedure.
Augmentation is important, because it is often desirable to reduce the size of the alphabet,
and thus
the cost of practical implementations. The following design is an example of augmentation.
\end{remark}

\begin{design} \label{smile}
We use several basic arrays to compose a smile face shown in Fig. \ref{smile}.
The colors indicate different basic complementary array sets: The two green arrays (at the lower and upper boundaries) form a basic array set with parameters $(M,N,L)=(2,2,8)$. So do the brown ones (at the lower left and upper right boundaries) and cyan ones (at the lower right and upper left boundaries). The two blue arrays (at the lower and upper boundaries) form a basic array set with $(M,N,L)=(2,2,4)$. The two yellow arrays (single points at the lower and upper boundaries) form a basic array set with $(M,N,L)=(2,1,1)$. The three black arrays (the eyes and nose) form a basic array set with $(M,N,L)=(3,3,3)$ (Fig. \ref{triangleBasic}); so do the three red ones (part of the mouth). The four purple arrays (the rest part of the mouth) form a basic array set with parameters $(M,N,L)=(4,2,3)$, which could be obtained by applying Theorem \ref{growSchemeGeneralization} with Remark \ref{remark_chooseU} and $\tilde{M}=4$ to a set of three single-point arrays.
By applying Remark \ref{catDesign} and Theorem \ref{growScheme} to the 20 arrays, a smile design with $(M,N,L)=(20,60,88)$ could be obtained. Due to Remark \ref{augmentationRemark}, another smile design with $(M,N,L)=(32,6,88)$ could be obtained.

\end{design}
\begin{figure}[h]
  \centering
  \includegraphics[width=0.4\textwidth]{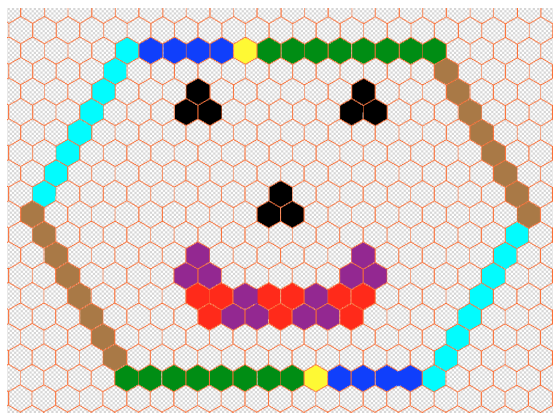}
  \caption{Smile design, with $(M,N,L)=(20,60,88)$ (without augmentation) or $(32,6,88)$ (with augmentation)}
  \label{smile}
\end{figure}



\subsection{Design for Infinitely Large Hexagonal Arrays} \label{sec_infinite}

By choosing a proper seed and growth scheme, we may be able to design infinitely large hexagonal arrays. The following is an example.
We first design a complementary array set with $M=7$, as a seed.

\begin{design} \label{7_6design}
The union of Design \ref{3_3design} with design parameters $(M,N,L) = (4,2,7)$ and Design \ref{4_2design} with $(M,N,L) = (3,3,7)$
is a design with $(M,N,L) = (4+3,{\rm lcm}(2,3),7) = (7,6,7)$ (Fig. \ref{7point}), based on Remark \ref{catDesign}.
\end{design}

\begin{design} \label{7InfiniteDesign}
By repeated application of Theorem \ref{growScheme} with $\bm U =\bm F_7$ to Design \ref{7_6design},
we obtain a design with parameters $(M,N,L)=(7,42,7^{\ell})$ for any positive integer $\ell$.
The design is illustrated in Fig. \ref{fig:largeDesign}, where the colors indicate the process  of ``growth''.
In fact, applying Theorem \ref{growScheme} to Design \ref{7_6design} once (with Remark \ref{remark_chooseF} conditions), we obtain a larger array set with $(M,N,L)=(7,{\rm lcm}(6,7), 7 \times 7)=(7,42,7^2)$. Fig. \ref{fig:largeDesign1} illustrates how the $7$ arrays of size $7$ (indicated by different colors) are combined to form larger arrays.
Similarly, applying Theorem \ref{growScheme} to the $(M,N,L)=(7,42,7^2)$ design once, we obtain a larger array set with $(M,N,L)=(7,{\rm lcm}(42,7), 7 \times 7^2)=(7,42,7^3)$. Fig. \ref{fig:largeDesign2} illustrates how the $7$ arrays of size $7^2$ (indicated by different colors) are combined  to form larger arrays. Further applications of Theorem \ref{growScheme} will not increase $M, N$, but will increase $L$.
\end{design}

As an alternative, the following design is also for $7^{\ell}$-point hexagonal array, but with different elements.

\begin{design} \label{7InfiniteDesign_alternative}
We keep applying Theorem \ref{growScheme} with $\bm U=\bm F_7$ to a single-point array, e.g. with entry $1$,
 we obtain a design with parameters $(M,N,L)=(7,7, 7^{\ell})$ for any positive integer $\ell$.
 To see how it works, first consider a complementary array set with $(M,N,L)=(1,1,1)$ (a single-point array). Taking the union of $7$ such array sets as in Remark \ref{catDesign} leads to an array set with $(M,N,L)=(7,1, 1)$.
Then, applying Theorem \ref{growScheme} once
  leads to an array set with $(M, N, L)=(7, {\rm lcm}(7,1), 7\times 1)=(7,7,7)$.
 Further applications of Theorem \ref{growScheme} will increase $L$, but not $M, N$.
The design could also be illustrated by Fig. \ref{fig:largeDesign}.
\end{design}

\begin{figure*}
     \begin{center}
        \subfigure[
        		 Illustration of one of the $7$ arrays of size $7^2$ obtained as a result of applying Theorem
		 \ref{growScheme} to a complementary array set of $7$ arrays of size $7$ (indicated by different colors) once
	 ]{
            \label{fig:largeDesign1}
            \includegraphics[width=0.2 \textwidth]{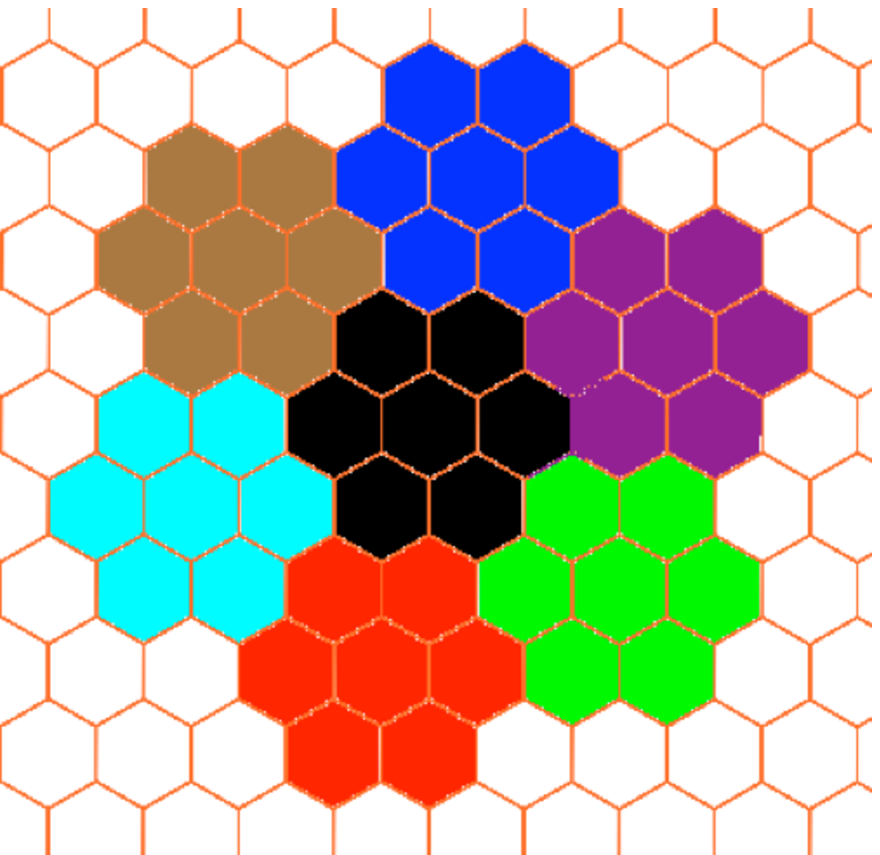}
        }%
        \hspace{1cm } 
        \subfigure[
        		Illustration of one of the $7$ arrays of size $7^3$ obtained as a result of applying Theorem
		 \ref{growScheme} to a complementary array set of $7$ arrays of size $7^2$ (indicated by different colors) once
	]{
            \label{fig:largeDesign2}
            \includegraphics[width=0.5 \textwidth]{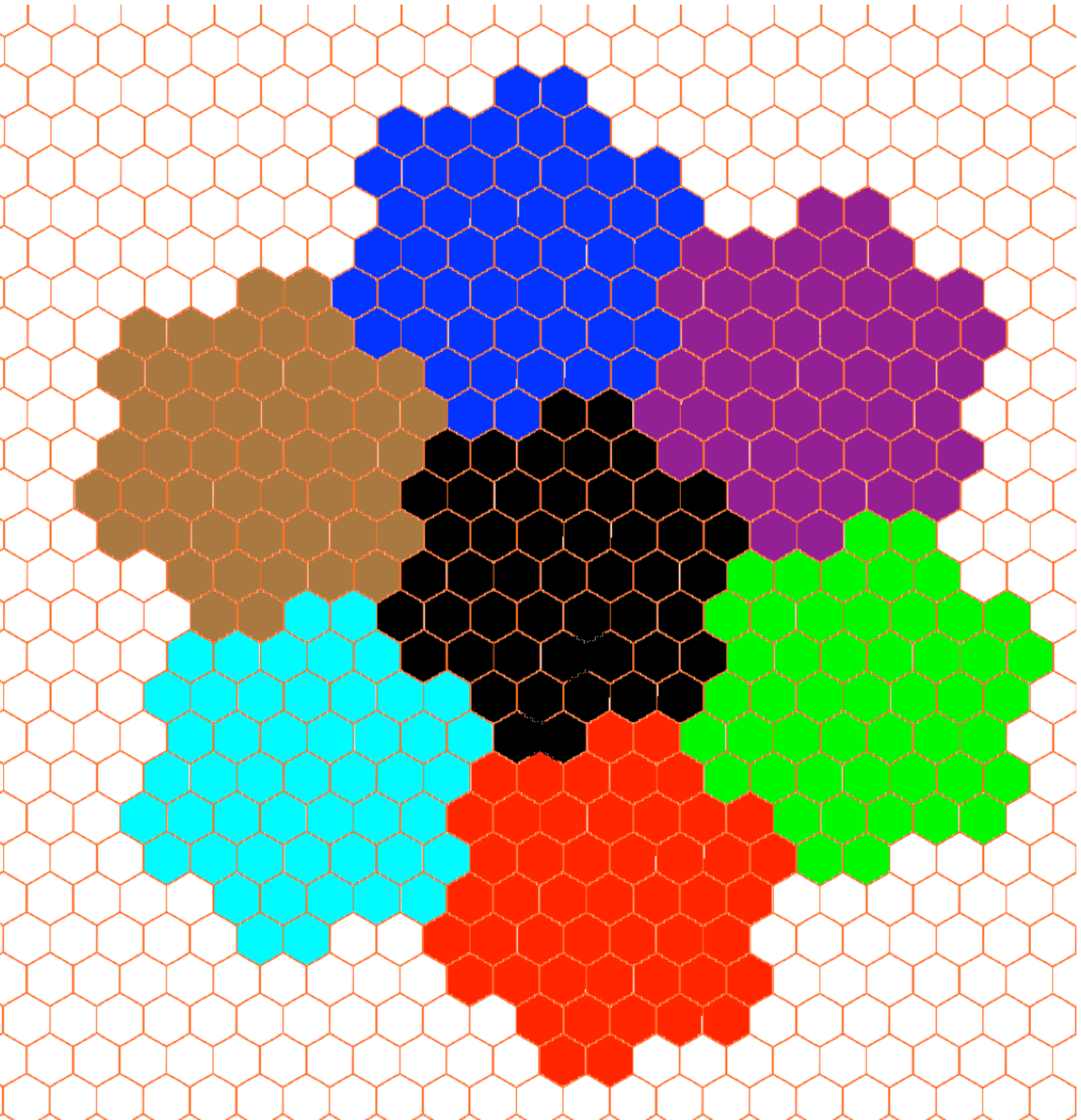}
        }
     \end{center}
    \caption{
      	 Complementary array set with parameter $(M,N)=(7,42, 7^{\ell})$ (Design \ref{7InfiniteDesign}),
  	or $(M,N)=(7,7,7^{\ell})$ (Design \ref{7InfiniteDesign_alternative}) for any positive integer $\ell$
     }
   \label{fig:largeDesign}
\end{figure*}

\section{URA, HURA and MURA Based on Periodic Autocorrelation} \label{secURA}

In this section, we review related works on URA, including HURA and MURA. 




\subsection{URA} \label{subURA}

We first introduce the concept of periodic autocorrelation and ``pseudo-noise''  that are important to the design of URA.

\begin{definition} \label{periodicDef}
Let $C=\{C[i_1,\cdots, i_n]\}$ be an \textbf{infinite array} on an integer lattice,
which satisfies
\begin{align*}
&C[i_1, \cdots, i_n]= C[I_1, \cdots, I_n], \quad \forall \, i_1, \cdots, i_n  \in \mathbb{Z},
i_1 \equiv I_1 \textrm{ (mod } L_1), \cdots,  i_n \equiv I_n \textrm{ (mod } L_n),  \\
& L_1, \cdots, L_n, \,  I_1, \cdots, I_n \in \mathbb{N},  \,  (I_1, \cdots, I_n) \in [0,L_1-1] \times \cdots \times [0,L_n-1].
\end{align*}
The finite array within $[0,L_1-1] \times \cdots \times [0,L_n-1]$, denoted by $c$, is called the basic array.
 $C$ is called the periodic extension of $s$.
The \textbf{periodic autocorrelation} function of $C$ (or $c$) is
\begin{align*}
A^C(v_1,\cdots , v_n) = \sum_{\substack{i_1 \in [0,L_1-1] , \cdots, i_n \in [0,L_n-1]} }
C[i_1, \cdots, i_n] \overline{C[i_1+v_1, \cdots, i_n+v_n]},  \quad v_1, \cdots,v_n \in \mathbb{Z}.
\end{align*}
The periodic crosscorrelation between two arrays are similarly defined.
\end{definition}

A section of an infinite array $C$ would be a valid URA aperture, if there exists a finite array $D$ such that
$C*D^{-}$ is a periodic extension of the  discrete delta function.
For a detailed discussion about the benefits and implementations of periodic extension, please refer to \cite{fenimoreCAI_URA1978}. 

\begin{definition} \label{PNDef}
An array of size $L_1 \times\cdots \times L_n$ is a \textbf{pseudo-noise} (PN) array if

(1) it is $\{\pm 1\}$-binary;

(2) all out-of-phase correlations are $-1$ \cite{key2}, i.e.
\begin{align*}
A^C(v_1,\cdots,v_n) &= -1,  \quad v_1,\cdots, v_n  \in \mathbb{Z}, v_1 \not \equiv 0 \textrm{ (mod } L_1), \cdots, \textrm{ or } v_n \not \equiv 0 \textrm{ (mod } L_n) . 
\end{align*}
\end{definition}

In 1967, Calabro and Wolf \cite{key2} showed that a class of two-dimensional PN arrays
 could be synthesized from quadratic residues.
The arrays are of size $p_1 \times p_2$, where $p_1, p_2$ are any prime numbers
satisfying $p_2 - p_1 = 2$ 
\begin{align*} 
D[i_1,i_2]=\left\{
\begin{array}{lcl}
-1       &      & {i_2 \equiv 0 \textrm{ mod } p_2} \\
1     &      & {i_1 \equiv 0  \textrm{ mod } p_1, \, i_2 \not \equiv 0  \textrm{ mod } p_2}\\
(i_1 / p_1)(i_2 / p_2)    &      &{\textrm{otherwise}}
\end{array} \right.
\end{align*}
where
$(i/p),  i \in \mathbb{Z}$ is Legendre operator:
$$ (i/p)=\left\{
\begin{array}{lcl}
0       &      & {i  \equiv 0 \textrm{ mod } p} \\
1     &      & {\exists  \, x \not \equiv 0 \textrm{ mod } p, \textrm{ s.t. } i \equiv x^2 \textrm{ mod } p}\\
-1       &      & {\textrm{otherwise.}}
\end{array} \right. $$

In 1978, following from the above result, Cannon and Fenimore \cite{fenimoreCAI_URA1978} designed $C$ and $D$ such that
$C*D^{-}$ is a periodic extension of the discrete delta function.
The design is given below,
where $p_1,p_2$ is a twin prime pair.
The coding array is $C$:
\begin{align} \label{URA_C}
 C[i_1,i_2]=\left\{
\begin{array}{lcl}
1    &      & { (i_1 / p_1)(i_2 / p_2)=1 }\\
0       &      & {i_2 \equiv 0 \textrm{ mod } p_2} \\
1     &      & {i_1 \equiv 0  \textrm{ mod } p_1, \, i_2 \not \equiv 0  \textrm{ mod } p_2}\\
0       &      & {\textrm{otherwise.}}
\end{array} \right.
\end{align}
The decoding array of $C$ is $D^{-}$, where $D$ is:
\begin{align*} 
D[i_1,i_2] = \left\{
\begin{array}{lcl}
1     &      & {\textrm{if }  C[i_1,i_2] = 1}\\
-1       &    & {\textrm{if }  C[i_1,i_2] = 0.}
\end{array} \right.
\end{align*}
It is shown that
\begin{align} \label{e11_final}
C*D^{-} = \frac{p_1 p_2-1}{2}\delta[\bm r] \textrm{ (within one period)}.
\end{align}

URA may also be designed from Maximal-length Shift-register Sequences
or m-sequences \cite{fenimore1981fast}.
m-sequence is another class of PN sequences,
which have lengths $n = 2^k - 1$ with $k$ being any positive integer.
They are sometimes referred to as ``PN sequences''
or ``m-sequences'' \cite{1976PN}.
In 1976, MacWilliams and Sloane showed how to obtain PN arrays
from m-sequences \cite{1976PN}.
Let $S$ be an m-sequence of length $n=2^k-1$.
If $ n = n_1 n_2 $ such that $n_1$ and $n_2$ are relatively prime,
a PN array $H$ is designed below:
\begin{align} \label{map_via_diagonal}
&H[i_1,i_2] = S[i], \textrm{ where }
i \equiv i_1 \textrm{ mod } n_1 , \, 0 \leq i_1 < n_1,    \textrm{ and }
i \equiv i_2 \textrm{ mod } n_2 , \, 0 \leq i_2 < n_2.
\end{align}
We note that when sum$(H)=-1$ (\cite{1976PN}, Property IP-$\uppercase\expandafter{\romannumeral4}^{*}$),
 we can design a URA with coding array $C=(-H+J)/2$ and decoding array $D=-H^{-}$ on integer lattices, where 
$J$ is a unit array, i.e. with all elements equal to one.

\subsection{HURA} \label{subHURA}

In 1985, Finger and Prince \cite{HURA1985hexagonal} designed a class of linear URA, i.e. sequences $C$ and $D$
such that $C*D$ is a periodic extension of the discrete delta function. The design is based on PN
sequences which in turn come from quadratic residues.
Then, by mapping linear sequences onto hexagonal lattice, they proposed the
hexagonal uniformly redundant arrays (HURA).
In the first step,
they  constructed  the following sequence of length $p$, where $p \equiv 3 \textrm{ mod }4$ is a prime:
\begin{align*}
D[i] = \left\{
\begin{array}{lcl}
1     &      & {\textrm{if }  i = 0}\\
-(i/p)       &    & {\textrm{otherwise} }
\end{array} \right.
\end{align*}
Let $C = (D+J)/2$. Then the following identity holds: 
\begin{align} \label{e10_final}
C*D^{-}  = \frac{p+1}{2}\delta[\bm r] \textrm{ (within one period)}.
\end{align}
In the second step, they map the sequence $D$ onto a hexagonal lattice:
\begin{align} \label{map_via_row}
H[i_1 \bm e_1 + i_2 \bm e_2] = D[i_1 + \tau i_2],
\end{align}
where $\tau$ is an integer to be chosen. $H$ is called the Skew-Hadamard URA.
It is easy to see that the correlation between $H$ and $(H+J)/2$ is a multiple of the discrete delta function,
just like the one-dimensional case.

As to the choice of lattice and $\tau$, it is well stated in \cite{HURA1985hexagonal} that
``The freedom available in this procedure rests in the choice of the lattice, the choice of
the order $p$, and the choice of the multiplier $\tau$. The lattice type will determine what symmetries
can occur $\cdots$ The multiplier $\tau$ determines the periods of the URA and hence the shape of the basic pattern.''
Furthermore, HURA are those with hexagonal basic patterns, when the lattice is chosen to be
hexagonal.
The qualified $p$ is either $3$ or primes of the form $12 k +1$ \cite{HURA1985hexagonal}.

Besides the fact that HURA are based on hexagonal lattices,
they are antisymmetric upon 60 degree rotation.
This property provides for effective reduction of
 background noise  \cite{cook1984gamma,CAI_astronomy1987}.
 Due to similar reasoning, the designs proposed
 in Section \ref{secTheory} also obtain robustness against background noise.

\subsection{MURA} \label{subMURA}

It has been shown that PN sequences, together with the  URA and
 HURA that are based on them,
could be made with prime lengths of the form  $4k + 3$.
Gottesman and Fenimore \cite{MURA1989_linear_square_hex} proposed the modified uniformly redundant arrays (MURA),
which further increased the available patterns for CAI.
MURA exist in lengths $p=4k+1$ where $p$ is a prime.

The design of MURA also starts with a sequence $D$ which is then
 mapped onto a hexagonal lattice, following the same procedure as HURA.
Recalling URA and HURA designs from Subsections \ref{subURA} and \ref{subHURA}, we may design using the procedure below:

Step 1. Let $D$ be a PN sequence (array);

Step 2. Let the coding array $C$ be $(D+J)/2$, and the decoding array be $D^{-}$;

Step 3. (optional) We map sequences onto a two-dimensional lattice (see Equations  (\ref{map_via_diagonal}) and (\ref{map_via_row})).

However, the design of MURAs is less straightforward, because $D$ is not a PN sequence and $C \neq (D+J)/2$. 
One way to design MURA sequences is:
\begin{align*}
C[i]&=\left\{
\begin{array}{lcl}
0       &      & {i  \equiv 0 \textrm{ mod } p} \\
1     &      & {\exists  \, x \not \equiv 0 \textrm{ mod } p, \textrm{ s.t. } i \equiv x^2 \textrm{ mod } p}\\
0       &      & {\textrm{otherwise}}
\end{array} \right.
\\
D[i]&=\left\{
\begin{array}{lcl}
1       &      & {i  \equiv 0 \textrm{ mod } p} \\
1     &      & {C[i] = 1, i \not \equiv 0 \textrm{ mod } p }\\
-1       &      & {\textrm{otherwise.}}
\end{array} \right.
 \end{align*}
It is easy to verify that for any $v \not \equiv 0 \textrm{ mod } p$, we have
$\sum_{i=0}^{p-1} C[i] D[i+v] = 0.$

Gottesman and Fenimore also gave a class of MURA for integer lattices.
The coding array is the same as (\ref{URA_C}), except for a change of the size: $p_1 = p_2 = p$.
The decoding array is $D^{-}$, where
\begin{align*}
D[i_1,i_2]&=\left\{
\begin{array}{lcl}
1       &      & {i_1+i_2  \equiv 0 \textrm{ mod } p} \\
1     &      & {C[i_1,i_2] = 1, i_1+i_2  \equiv 0 \textrm{ mod } p}\\
-1       &      & {\textrm{otherwise.}}
\end{array} \right.
\end{align*}
\section{New URA Constructions} \label{newPeriodic}

%
%
%
%
%

\subsection{URA from Periodic Complementary Sequence Set} \label{newPeriodic3}

In this section, we first briefly summarize some similarities and differences between the aperiodic-based and periodic-based designs of CAI,
and then propose a new design framework that is based on periodic autocorrelation. 

In the aperiodic case, the elements of arrays are assumed to extend only over some finite area and be zero outside that area.
This fact provides great convenience for the design of complementary array sets/banks, since several arrays could be easily concatenated while maintaining the unimodular alphabet during the ``growth'' process.
In addition, the concept of ``bank'' and a growth scheme make the aperiodic-based designs more flexible. 
For example, we have shown how to make CAI aperture with arbitrary patterns.
In the periodic case, the arrays were assumed to be periodic and infinite in extent.
The resulting correlations are calculated over a full period.
The usual way to design is to first design sequences with good autocorrelation property, e.g. pseudo-noise, and then map them onto arrays.
As to practical implementations, periodic-based designs often require the physical coding aperture to be periodic extensions of some basic patterns to mimic the periodicity, while aperiodic-based ones do not.

Despite their differences in principles and implementations, the idea of ``complementary'' can also be associated with periodic correlations, 
leading to the following concept that is similar to complementary array sets in Section \ref{secTheory}.


\begin{definition}
A set of arrays with the same basic pattern is a periodic complementary array set (PCAS), if the sum of
their periodic autocorrelations is a periodic extension of the discrete delta function.
A one-dimensional PCAS is also referred to as a periodic complementary sequence set (PCSS).
The notation ``design parameters'' $(M,N)$ or $(M,N, L)$ is similarly defined as in Subsection \ref{7pointSection}.

\end{definition}

As discussed before, URA (including HURA) require the lengths of sequences to be prime numbers or $2^k-1$, so the possible sizes of URA arrays are quite limited.
However, the above concept produces more admissible lengths, offering more choices in selecting an aperture.
For example, we can construct the following URA sequence  of length $6$. 

{\bf Example.}  
PCSS with parameter $(M,N,L) = (4,2,6)$:
\begin{align*}
&S_1 = \{1,    -1,    -1,    -1,    -1,    -1\}, \quad
S_2 = \{1,     1,     1,     1,    -1,    -1 \}, \quad
S_3 = \{-1,    1,    -1,     1 ,   -1  ,  -1\}, \quad
S_4 = \{-1,    -1,     1 ,    1  ,  -1 ,    1\}.
\end{align*}
Then, the sequences may be mapped onto a two-dimensional lattice, following  procedures similar
to Equations (\ref{map_via_diagonal}) and (\ref{map_via_row}). 
Now a natural question that arises is:
for a given alphabet, what are the possible lengths for which there exists a PCSS? and how to design them?
This will be addressed in the remaining sections.

 A natural way to construct PCSS is to synthesize them from existing designs.
 Some synthesis methods have been provided for binary PCSS in \cite{bomer1990periodic}, and they
 could be easily extended to the non-binary case. In the following two sections, we propose some different
 synthesis methods.
%
%

At the end of this subsection, there are two remarks worth mentioning.
First, the concept  of PCSS is not new. It was once referred to as ``periodic complementary sequences'' or ``periodic complementary binary sequences'' \cite{bomer1990periodic}.
To the best of our knowledge, prior works mainly focused on the binary case. One possible reason is its intimate
relationship with cyclic difference sets. 
Second, complementary sequence sets are subclasses of PCSS due to the following fact:
\begin{align} \label{ACS_PCS}
A_p^S(v) = A_a^S(v) + A_a^S(v-L) \quad \forall \, v  \in \mathbb{Z},  0 \leq v < L,
\end{align}
where $S$ is a sequence of length $L$, $A_p(\cdot), A_a(\cdot)$ respectively denote periodic, aperiodic autocorrelations.

\subsection{Synthesis Methods from the Chinese Remainder Theorem} \label{secSynthesis_CRT}

\subsubsection{PCAS synthesized from PCSS and perfect sequence} \label{C1}

A sequence  is called a ``perfect sequence'' if its periodic autocorrelation is a periodic extension of the discrete delta function.
Consider a PCSS $\{S_m\}_{m=1}^M$ of length $s$, and a perfect sequence $S$ of length $t$.
We can then construct a PCAS $\{C_m\}_{m=1}^M$ of size $s \times t$ (or similarly $t \times s$):
\begin{align} \label{hello1}
C_m[i,j]= S_m[i] S[j],  \quad m=1,\cdots,M, \, i, j  \in \mathbb{Z}.
\end{align}

\begin{proof}
The periodic autocorrelation of $C_m$ satisfies
$$A^{C_m} (v_1,v_2) = A^{S_m}(v_1) A^{S}(v_2).$$
Thus, for any $v_1, v_2  \in \mathbb{Z},  (v_1,v_2)\neq (0,0)$,
$$
\sum_{m=1}^M A^{C_m} (v_1,v_2)  = \left( \sum_{m=1}^M A^{S_m}(v_1) \right)A^{S}(v_2) = 0.
$$
\end{proof}

\subsubsection{ PCSS synthesized from PCSS and perfect sequence} \label{C2}

Consider a PCSS $\{S_m\}_{m=1}^M$ of length $s$, and a perfect sequence $S$ of length $t$.
Also assume that $s$ and $t$ are co-prime.
We can then construct a PCSS $\{\tilde{S}_m\}_{m=1}^M$ of length $s t$:
\begin{equation} \label{map}
\tilde{S}_m[i]= C_m[i \textrm{ mod } s, i \textrm{ mod } t], \quad  i  \in \mathbb{Z},
\end{equation}
where $\{C_m\}$ is given in Subsection \ref{C1}.

\begin{proof}
Equation (\ref{map}) provides a one-to-one mapping between a sequence and an array, guaranteed by the
Chinese Remainder Theorem. The mapping is linear
so that the autocorrelation function is preserved, i.e.
$$A^{\tilde{S}_m} (v) = A^{C_m} (v \textrm{ mod } s, v \textrm{ mod } t),$$
and thus the sequence set $\{\tilde{S}_m\}_{m=1}^M$ is complementary.
\end{proof}

\subsubsection{PCSS/PCAS from two PCSS with co-prime lengths} \label{C3}

Consider a PCSS $\{S_{m_1}\}_{m_1=1}^{M_1}$ of length $s$, and another PCSS $\{T_{m_2}\}_{m_2=1}^{M_2}$ of length $t$.
We can then construct a PCAS $\{C_{(m_1,m_2)}\}$ 
of size $s \times t$:
\begin{align} \label{last_eq1}
C_{(m_1,m_2)}[i,j] = S_{m_1}[i] T_{m_2}[j], \quad m_1=1,\cdots,M_1, m_2=1,\cdots,M_2, \, i,j  \in \mathbb{Z}.
\end{align}
Further, if $s$ and $t$ are co-prime, we can construct a PCSS of length $st$.

\begin{proof}
For a given $1 \leq m_2 \leq M_2$,
\begin{align*}
 \sum_{m_1=1}^{M_1} A^{C_{(m_1,m_2)}} (v_1,v_2) &= \left\{
 \begin{array}{lcl}
  s \cdot A^{T_{m_2}}(v_2) & &{v_1 = 0} \\
  0 & &{\textrm{otherwise}} \\
 \end{array} \right.
 \end{align*}

Thus, for $v_1, v_2  \in \mathbb{Z}, (v_1,v_2) \neq (0,0)$,
 \begin{align*}
\sum_{m_2=1}^{M_2}\sum_{m_1=1}^{M_1} A^{C_{(m_1,m_2)}} (v_1,v_2) =& \sum_{m_2=1}^{M_2} \left( \sum_{m_1=1}^{M_1} A^{C_{(m_1,m_2)}} (v_1,v_2) \right)
 	= 0
\end{align*}
If $s$ and $t$ are co-prime, a PCSS could be designed using the mapping given in Equation (\ref{map}).
\end{proof}

\begin{remark}
This result is stronger than that given in \cite{bomer1990periodic} (Theorem 6), since it does not require the
number of sequences to be relatively prime.
\end{remark}

\subsubsection{PCAS constructed from another PCAS of a different size} \label{C4}

Assume that we have a PCAS of size $s \times t$ synthesized from PCSS $\{S_{m_1}\}_{m_1=1}^{M_1}$ and $\{T_{m_2}\}_{m_2=1}^{M_2}$ using the method in Subsection \ref{C3}. 
Suppose that $\textrm{gcd}(s,t) \neq 1$, but $s=s_1 s_2$ for some $s_1 \neq 1$ and $ s_2 \neq 1$ where $\textrm{gcd}(s_1,s_2) = \textrm{gcd}(s_2,t)=1$.
A PCAS of size $s_1 \times s_2 t$ could be designed by 
first mapping the PCSS $\{S_{m_1}\}_{m_1=1}^{M_1}$ to a PCAS of size $s_1 \times s_2$ in a way similar to Equation (\ref{hello1}), then constructing  a three-dimensional PCAS of size $s_1 \times s_2 \times t$ in a way similar to Equation (\ref{hello1}), and finally 
 mapping the latter two dimensions to a single dimension in a way similar to Equation (\ref{map}), resulting in a PCAS of size $s_1 \times s_2 t$.


\subsection{Synthesis via Unitary Matrices} \label{secSynthesis_Unitary}

\begin{theorem} \label{syn_unitary}
For any positive integer $ s $, there exists at least one PCSS with design parameters
$(M, N, L) = (p_n, \, p_1 \cdots p_n, s)$,
where  $p_1 < \cdots < p_n$ are all the distinct prime divisors of $s$.

For any positive integers $ s_1,\cdots, s_k $, there exists at least one PCAS of size $s_1 \times \cdots \times s_k$
with design parameters
$(M,N,L) = (p_n, p_1 \cdots p_n,st)$,
where  $p_1 < \cdots < p_n$ are all the distinct prime divisors of $s_1  \cdots  s_k$.
\end{theorem}

\begin{proof}
We prove the first part constructively.
Using Equation (\ref{ACS_PCS}), we observe that it suffices to construct an aperiodic complementary array set. 
Without loss of generality, assume that $s = p_1^{q_1} \times p_2^{q_2} \cdots \times p_n^{q_n}$
is a prime factorization of $s$, where $q_j \geq 1$ and $p_j$'s are distinct for $j=1,\cdots,n$.
Let $ S^{(0)} =  \{S_m^{(0)} \}_{m=1}^{p_1}$ be a set of $p_1$ sequences each of which contains a single point $1$, i.e.
$S^{(0)}$ has design parameters $(M, N, L)=(p_1,1,1)$.

In the first iteration, we apply Theorem \ref{growScheme} to $S^{(0)}$ with $\bm U$ equal to the Fourier matrix $\bm F_{p_1}$ while satisfying Remark \ref{remark_chooseF} conditions to obtain a (one-dimensional) complementary array set with
$(M,N,L)=(p_1, {\rm lcm}(p_1,1), p_1 \times 1 )=(p_1,p_1,p_1)$.
Applying Theorem 3 a second time, we obtain a complementary array set with 
$(M, N, L)= (p_1, {\rm lcm}(p_1,p_1), p_1 \times p_1 )=(p_1,p_1,p_1^2)$.
After applying Theorem 3 to $S^{(0)}$ $q_1-1$ times, we obtain the complementary array set $S^{(1)} $ with $(M, N, L)=(p_1,p_1,p_1^{q_1-1})$.

In the second iteration, we first apply Theorem \ref{growSchemeGeneralization} to $S^{(1)}$ with $\tilde{M}=p_2$ while satisfying Remark \ref{remark_chooseU} conditions. The resulting complementary array set has parameters
$(M, N, L)=(p_2, {\rm lcm}(p_1,p_2), p_1 \times p_1^{q_1-1} )=(p_2, p_1p_2, p_1^{q_1})$;
then we apply Theorem 3 $q_2-1$ times with $\bm U$ being the Fourier matrix $\bm F_{p_2}$ to create the complementary array set $S^{(2)}$ with
$(M, N, L)=(p_2, p_1p_2, p_1^{q_1} p_2^{q_2-1})$.

By recursive construction as above, after $w$th iteration we obtain the complementary array set $S^{(w)}$ with 
$(M,N,L) = (p_n, p_1 \cdots p_n, p_1^{q_1} p_2^{q_2} \cdots p_n^{q_n-1})$.
Finally, applying Theorem \ref{growScheme} with $\bm U$ equal to the Fourier matrix $\bm F_{p_n}$ an extra time to $S^{(w)}$,  we obtain a complementary array set with $(M,N,L) = (p_n, p_1 p_2\cdots p_n, p_1^{q_1} p_2^{q_2} \cdots p_n^{q_n})$.

The proof of the second part is similar.

\end{proof}

\begin{remark}
Theorem \ref{syn_unitary}  gives the construction for PCAS of an arbitrary size,
where the alphabet is determined by the product of the distinct prime divisors of the size, and the
order is the largest  prime divisor.
From the proof of Theorem \ref{syn_unitary}, the result also holds for aperiodic autocorrelations. 

A natural question that arises is how tight the result in Theorem \ref{syn_unitary} is. Specifically, is there any solution whose order $N$ is less than $p_n$?
This is clearly not the case for Golay complementary sequence pair, where the size $s$ is a power of $2$. Although we were not able to answer this question in general, we were able to prove the following results. 

\end{remark}

\begin{theorem} \label{complexConjecture}
For any prime number $p$,  a $p$-regular set is defined to be a set of $p$ distinct  unimodular  complex numbers
that form the vertices of a uniform polygon in the complex plane. 
Let $N = p_1^{r_1} \cdots p_n^{r_n}, r_j \geq 1$ be a positive integer with distinct prime divisors $p_j , j =1, \cdots , n$.
Consider $M$ variables $x_1, \cdots, x_{M}$ that take values in the set of $N$th root of unity. Suppose that $\sum_{m=1}^M x_m = 0$.
\begin{enumerate}
\item If $n \leq 2$, the set $ \{x_m\}_{m=1}^M$ can be written as the unions of
$p_k$-regular configurations, i.e.
\begin{align} \label{time10}
\{x_1 , \cdots, x_M\}  = \bigcup_{ \{(k, j) \mid c_k > 0, \,  k=1,2 , \, j =1,\cdots, c_k\} }  RC_j^{(p_k)}.
\end{align}

\item $M$ can be written as
\begin{align} \label{Thm3_2}
M = c_1 p_1 + c_2 p_2 , \quad c_1,c_2 \in \mathbb{N}\cup \{0\}.
\end{align}
\end{enumerate}
%
%
\end{theorem}

\begin{remark}
Consider an aperiodic complementary array set $\{S_m\}_{m=1}^M$ with design parameters $(M,N,L)$. 
Then we have $\sum_{m=1}^{M} S_m[0] S_m[L-1] = 0$. 
Assume that the alphabet is $N$-phase, where $N=p_1^{r_1}$ ($n=1$) or $N=p_1^{r_1} p_2^{r_2}$ with $p_1$ and $p_2$ distinct primes ($n=2$). 
Applying Theorem \ref{complexConjecture}, Equation (\ref{Thm3_2}) implies that 
(1) $M \geq p_1$ if $n=1$; (2) $M \geq \min\{p_1, p_2\}$ if $n=2$.
Due to similar reasons, $M=7$ in Design \ref{7InfiniteDesign_alternative} is tight whenever $N$ is a power of $7$.
\end{remark}

\section{Simulation Results} \label{secSim}

We have performed computer simulations  to demonstrate a multi-channel
CAI system and a classical URA-based one.

The multi-channel CAI system that we select comes from Design \ref{4_2design}. 
Admittedly, in practice we only need four pairs of aperture arrays with $\{ -1,1\}$-alphabet. 
But the coded images contain negative entries, which are not straightforward to illustrate by simulation (we used Matlab software). 
We thus provide an alternative approach which relies on the following lemma.  

\begin{lemma} \label{binaryMask}
Suppose that $\{(C_m, D_m)\}_{m=1}^M$ is a complementary bank with  alphabet  $\mathfrak{A} = \{ -1,1\}$ and it satisfies
$ \sum_{m=1}^M D_m =  0.$
 Then
$\{(\tilde{C}_m, D_m)\}_{m=1}^M$ is a complementary bank, 
where 
$\tilde{C}_m = (C_m + J)/2, \, m=1,\cdots,M$.
Here, $0$ and  $J$  are respectively the array of zeros and the array of ones, whose supports are the same as $D_m$.
\end{lemma}
\begin{proof}
The proof follows immediately from
\begin{align}
\sum_{m=1}^M \tilde{C}_m * D_m^{-} 
&=  \sum_{m=1}^M \frac{1}{2} (C_m + J) * D_m^{-} 
=  \frac{1}{2}  \sum_{m=1}^M C_m * D_m + \frac{1}{2} J * \left(\sum_{m=1}^M D_m \right)^{-} 
 = \frac{1}{2}  \sum_{m=1}^M C_m * D_m .
\end{align}
\end{proof}
The above result gives a general method to design a mask with simple closing/opening pinholes (the elements of $C$ are either $0$ or $1$).
In practice, the method may be of some interest on its own right, but we do not elaborate here. 
As a corollary of Lemma~\ref{binaryMask}, it is easy to see that if $\{C_m\}_{m=1}^M$ is a  complementary array set with  alphabet  $\mathfrak{A} = \{ -1,1\}$, then 
$$
\biggl\{ \biggl(  \frac{1}{2}(C_m+J), C_m \biggr) \biggr\}_{m=1}^M \bigcup
\biggl\{ \biggl( \frac{1}{2}(-C_m+J), -C_m \biggr) \biggr\}_{m=1}^M
$$
is a complementary bank.

Following Design \ref{4_2design} and the above result, we obtain the following $8$-channels CAI  $\{(C_m, D_m)\}_{m=1}^8$, each with a mask as shown in Figure \ref{7point}.
The coding arrays are
\begin{align*} 
\{ C_1[k] \}_{k=0}^6 &= \{ 1, 1, 0, 1, 1, 0, 1 \}, \quad
\{ C_2[k] \}_{k=0}^6 = \{ 0, 1, 1, 0, 0, 1, 1 \}, \quad
\{ C_3[k] \}_{k=0}^6 = \{ 1, 1, 0, 1, 0, 1, 1 \}, \\
\{ C_4[k] \}_{k=0}^6 &= \{ 1, 1, 1, 0, 1, 0, 1 \}, \quad
\{ C_5[k] \}_{k=0}^6 = \{ 0, 0, 1, 0, 0, 1, 0 \}, \quad
\{ C_6[k] \}_{k=0}^6 = \{ 1, 0, 0, 1, 1, 0, 0 \}, \\
\{ C_7[k] \}_{k=0}^6 &= \{ 0, 0, 1, 0, 1, 0, 0 \}, \quad
\{ C_8[k] \}_{k=0}^6 = \{ 0, 0, 0, 1, 0, 1, 0 \}.
\end{align*}
If we choose the element labeled $0$ to be the origin, the decoding arrays are
\begin{align*} 
\{ D_1[k] \}_{k=0}^6 &= \{ 1,  1, -1, 1, 1, -1, 1 \}, \quad
\{ D_2[k] \}_{k=0}^6 = \{ -1, -1, 1, 1, 1, 1, -1 \}, \quad
\{ D_3[k] \}_{k=0}^6  = \{ 1, -1, 1, 1, 1, -1, 1 \}, \\
\{ D_4[k] \}_{k=0}^6 &= \{ 1, 1, -1, 1, 1, 1, -1 \}, \quad
\{ D_5[k] \}_{k=0}^6 = \{ -1, -1, 1, -1, -1, 1, -1 \}, \quad
\{ D_6[k] \}_{k=0}^6 = \{ 1, 1, -1, -1,  -1, -1, 1 \}, \\
\{ D_7[k] \}_{k=0}^6 &= \{ -1, 1, -1, -1,-1, 1, -1 \}, \quad
\{ D_8[k] \}_{k=0}^6 = \{ -1, -1, 1, -1, -1, -1, 1 \}.
\end{align*}
Fig.~\ref{multi_channel_demo}  illustrates how a source object is coded and decoded 
in a multi-channel system. 
The source object is a  $131 \times 131$ pixels ``camera man'' image and the distance between two pinholes is $60$ pixels. 
The grey level is normalized to be in the range of $[0, 255]$.
 
The second simulation is for URA (shown in Fig.~\ref{URA_demo}). It uses the same source object and following aperture:
\begin{align} \left[
  \begin{array}{lcccl}
      0   & 1   &  1  &   1   &  1 \\
      0   & 0   & 1   &  1   &  0 \\
      0   & 1   & 0   &  0   &  1 \\
  \end{array} \right],
\end{align}     
which is a $3 \times 5$ array that comes from the m-sequence of length $15$ (See Section \ref{subURA} for details).
In the implementation, we use the arrangement suggested in \cite{fenimoreCAI_URA1978}, i.e. a 
$6 \times 10$ aperture composed of a periodic extension of the basic $3 \times 5$ patterns, and a $3 \times 5$ 
decoding array.

\begin{figure*}
     \begin{center}
        \subfigure[
            The upper-left image, ``cameraman'', is the source image.
     From the upper-middle to the bottom-right, the images are the coded images from
   apertures $C_1,\cdots,C_8$ in each channel.
   ]{
            \label{encode}
            \includegraphics[width=0.9 \textwidth]{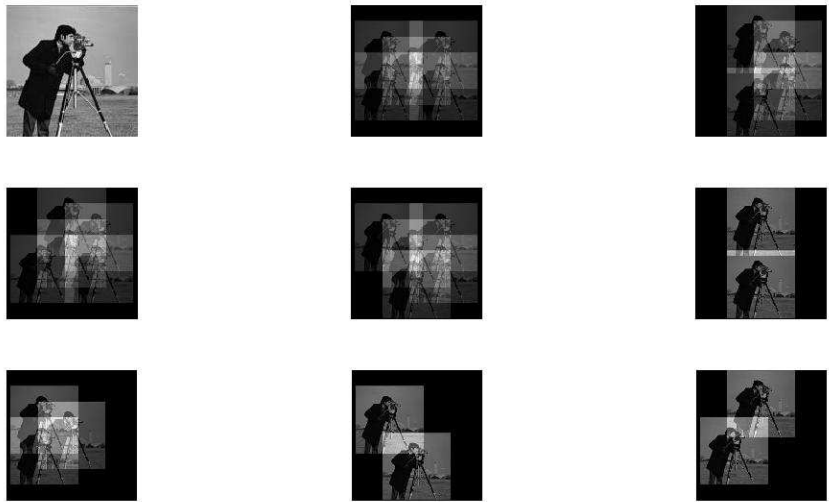}
        }%
        \\ 
        \subfigure[
            From the upper-left to the bottom-middle, the images are the decoded results from $D_1,\cdots,D_8$ in   each channel. The bottom-right image is the reconstructed image, coming from the addition of the eight    decoded results.
  ]{
            \label{decode}
            \includegraphics[width=0.9 \textwidth]{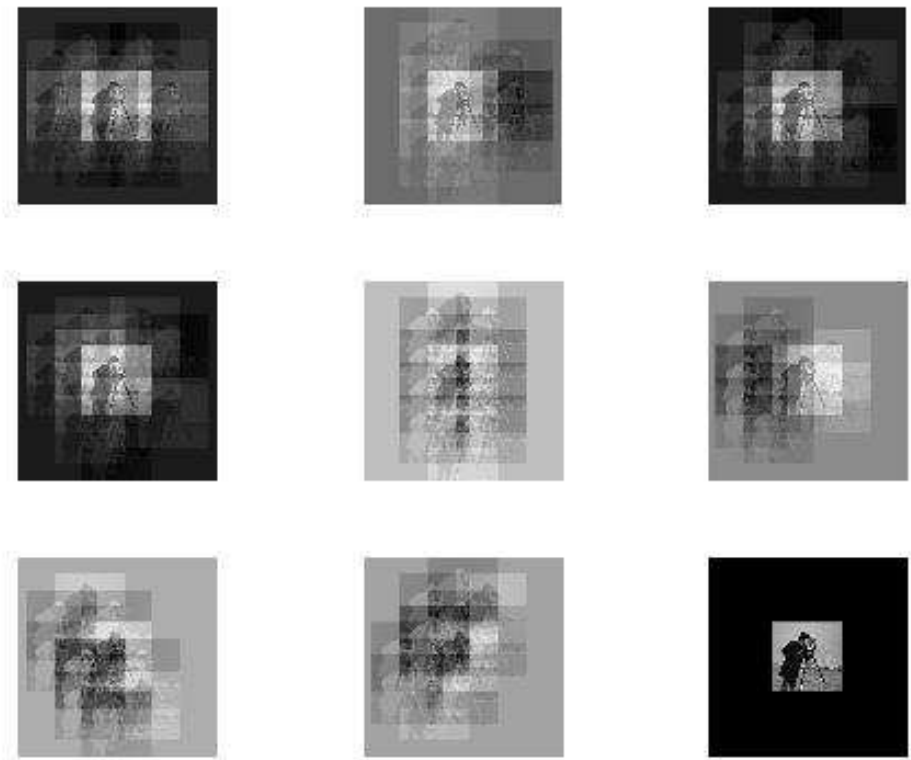}
        }
    \caption{
         Demonstration of the encoding and decoding process of a multi-channel CAI system
     }
     \label{multi_channel_demo}
     \end{center}
\end{figure*}

\begin{figure*}
     \begin{center}
        \subfigure[
            This figure shows the coded image of the URA system. The coded aperture produces a cyclic version
     of the basic aperture pattern.
   ]{
            \label{periodic_encode}
            \includegraphics[width=0.65 \textwidth]{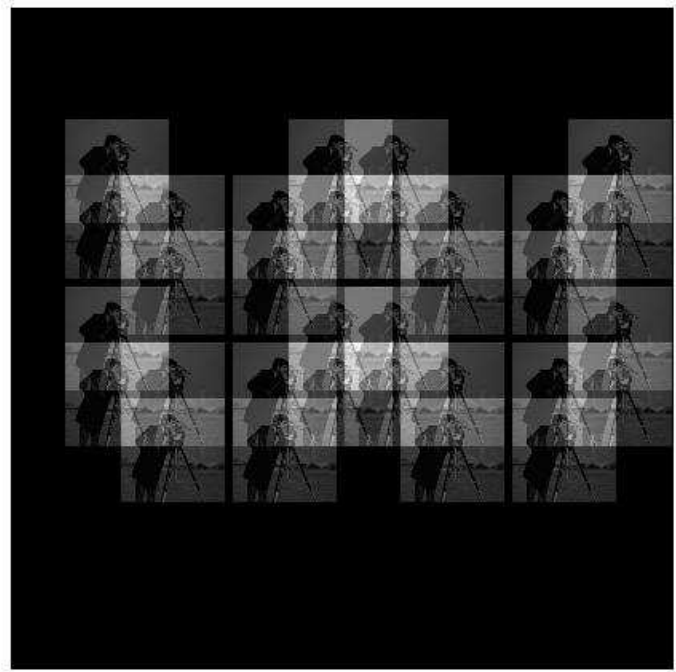}
        }%
        \\ 
        \subfigure[
            This figure shows the decoded image. The source image is reconstructed in the center area. 
  ]{
            \label{periodic_decode}
            \includegraphics[width=0.7 \textwidth]{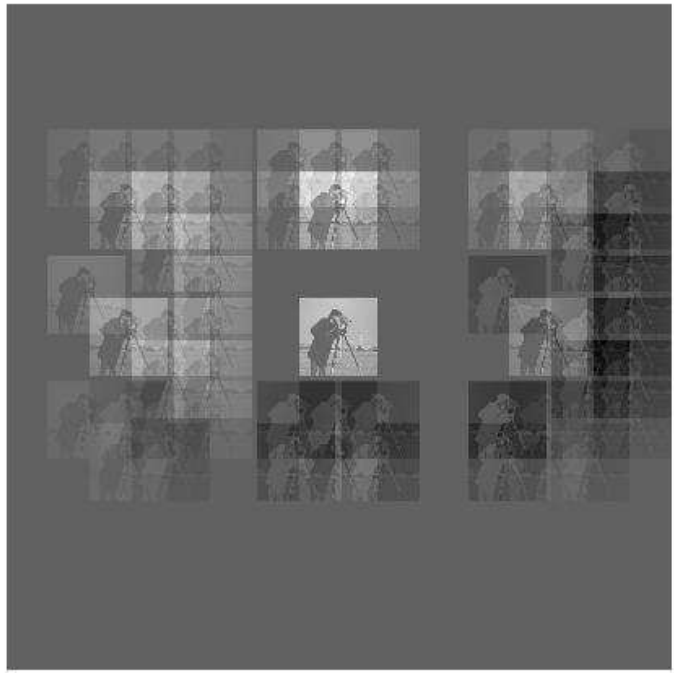}
        }
    \caption{
         Demonstration of the encoding and decoding process of a URA-based CAI system
     }
        \label{URA_demo}
       \end{center}
\end{figure*}

\appendices
\section{Proof of Theorem \ref{M2NonexistenceThm}} \label{appendix_a}

The following lemmas are helpful to the proof of Theorem \ref{M2NonexistenceThm}.

\begin{lemma} \label{4additionLemma}
If complex numbers $\alpha_1, \cdots, \alpha_{4}$ are unimodular and satisfy $\sum_{k=1}^{4} \alpha_k = 0$, then they contain $2$ opposite pairs.
\end{lemma}

\textit{Proof of Lemma \ref{4additionLemma}}:

Let
$y = \alpha_1 + \alpha_2, \, z = (- \alpha_3 )+(- \alpha_4),$
then $y=z$. Because $|\alpha_1| = |\alpha_2| = 1$, $y$ is on the bisector of $\alpha_1$ and $\alpha_2$. Similarly, $z$ is on the bisector of $-\alpha_3$ and $-\alpha_4$. Since $y=z$, We have $\alpha_1=-\alpha_3$ or $\alpha_1= -\alpha_4$.

\begin{lemma} \label{root_lemma2}
If roots of unity $\alpha_1,\alpha_2$ satisfy $|\alpha_1 - 2\alpha_2| = 1$, then $\alpha_1 = \alpha_2$.
\end{lemma}

\textit{Proof of Lemma \ref{root_lemma2}}:

The identity $1 = (\alpha_1-2\alpha_2)(\bar{\alpha_1}- 2 \bar{\alpha_2}) = 1 - 2 (\alpha_1 \bar{\alpha_2} + \bar{\alpha_1} \alpha_2 ) +4$ implies $\alpha_1 \bar{\alpha_2} = 1$, i.e. $\alpha_1=\alpha_2$.

\vspace{0.3cm}

\textit{Proof of Theorem \ref{M2NonexistenceThm}}:

Assume that there exists a complementary array pair. Writing down
(\ref{eqn11}) explicitly we obtain the following system of equations  consisting of $9$ equations and $14$ variables
$\{ x_k \}_{k=0}^6 \cup \{ y_k \}_{k=0}^6$.
%
\begin{align}
x_1 \bar{x_3} + x_6 \bar{x_4} + y_1 \bar{y_3} + y_6 \bar{y_4} &= 0 \label{e1} \\
x_0 \bar{x_1} +x_3 \bar{x_2} + x_5 \bar{x_6} + x_4 \bar{x_0} 
 		     + y_0 \bar{y_1} +y_3 \bar{y_2} + y_5 \bar{y_6} + y_4 \bar{y_0} &= 0 \label{e2} \\
x_1 \bar{x_4} +y_1 \bar{y_4} &= 0 \label{e3} \\
x_2 \bar{x_4} + x_1 \bar{x_5} + y_2 \bar{y_4} + y_1 \bar{y_5} &= 0 \label{e4} \\
x_0 \bar{x_2} +x_6 \bar{x_1} + x_4 \bar{x_3} + x_5 \bar{x_0} 
		     + y_0 \bar{y_2} +y_6 \bar{y_1} + y_4 \bar{y_3} + y_5 \bar{y_0} &= 0 \label{e5} \\
x_2 \bar{x_5} +y_2 \bar{y_5} &= 0 \label{e6} \\
x_3 \bar{x_5} + x_2 \bar{x_6} + y_3 \bar{y_5} + y_2 \bar{y_6} &= 0 \label{e7} \\
x_0 \bar{x_3} +x_1 \bar{x_2} + x_5 \bar{x_4} + x_6 \bar{x_0} 
		     + y_0 \bar{y_3} +y_1 \bar{y_2} + y_5 \bar{y_4} + y_6 \bar{y_0} &= 0 \label{e8} \\
x_3 \bar{x_6} +y_3 \bar{y_6} &= 0 \label{e9}
\end{align}
%
We only need to prove that the above  system of equations have no solution on the unit circle.
Assume without loss of generality that $x_1 = y_1 = 1$.
After simplifying Equations (\ref{e3}), (\ref{e1}), (\ref{e6}), (\ref{e9}) we have
\begin{align}
y_4 &= -x_4, \label{e3new} \\
y_3 &= -(x_3 + (\bar{x_6} - \bar{x_4} ) x_4) = -x_3 - \bar{x_6} x_4 +1,   \label{e1new} \\
y_5 &= - \bar{x_2} x_5 y_2,  \label{e6new} \\
y_6 &= - \bar{x_3} x_6 y_3.  \label{e9new}
\end{align}
From Equation  (\ref{e1new}), we obtain $y_3 + x_3 + \bar{x_6} x_4 - 1 = 0$. Due to Lemma \ref{4additionLemma}, we have three cases to consider:
\\
Case A: $y_3 = 1, x_4 = - x_3 x_6 $; \quad
Case B: $x_3 = 1, y_3 = - x_4 \bar{x_6} $; \quad
Case C: $x_6 = x_4, y_3 = - x_3$.

\subsection{Case A: $y_3 = 1, x_4 = - x_3 x_6 $}

From Equations  (\ref{e3new}), (\ref{e6new}), and (\ref{e9new}), we eliminate variables $x_4, y_3, y_4, y_5, y_6$ in Equations (\ref{e4}) and (\ref{e7}) and obtain
\begin{align}
 - x_2 \bar{x_3} \bar{x_6} + \bar{x_5} +  \bar{x_3}\bar{x_6} y_2- x_2 \bar{x_5} \bar{y_2} &= 0 \label{e4newnew} \\
 x_3 \bar{x_5} + x_2 \bar{x_6} - x_2 \bar{x_5} \bar{y_2} - x_3 \bar{x_6} y_2 &= 0 \label{e7newnew}.
\end{align}
We may write Equations (\ref{e4newnew}) and (\ref{e7newnew}) as
\begin{align}
x_5 \bar{x_3} (x_2 - y_2) &= x_6 (1-x_2 \bar{y_2}) \label{e1final} \\
x_5 (\bar{x_3} - \bar{x_2} y_2) &=  x_6 (\bar{x_3} \bar{y_2} - \bar{x_2} ) \label{e2final}.
\end{align}
If $x_2 \neq y_2$, Equation (\ref{e1final}) gives
$$
x_5 = x_3 \frac{1-x_2 \bar{y_2} }{x_2 - y_2} x_6 =  -x_3 \bar{y_2} x_6.
$$
If $x_2 \neq x_3  y_2$, Equation (\ref{e2final}) gives
$$
x_5 = \frac{\bar{x_3} \bar{y_2} - \bar{x_2} }{\bar{x_3} - \bar{x_2} y_2} x_6 = \bar{y_2} x_6.
$$
So there are four cases to consider that further eliminate the variables. 

\textit{Case A1:} \quad  $x_5 =  -x_3 \bar{y_2} x_6, \, x_5 =  \bar{y_2} x_6$

Clearly, $x_3 = -1$.  Eliminating variables in Equations (\ref{e2}) and (\ref{e8}) we obtain
\begin{align*}
(x_0 - \bar{x_2}  - \bar{x_2} -x_6 \bar{y_0}) + (\bar{y_2} + x_6 \bar{x_0} + y_0 + \bar{y_2}) &=0 \\ 
-(x_0 - \bar{x_2}  - \bar{x_2} -x_6 \bar{y_0}) +(\bar{y_2} + x_6 \bar{x_0} + y_0 + \bar{y_2}) &=0 
\end{align*}
which is equivalent to
\begin{align}
x_0 - \bar{x_2}  - \bar{x_2} -x_6 \bar{y_0} &= 0 \label{new_added1} \\
 \bar{y_2} + x_6 \bar{x_0} + y_0 + \bar{y_2} &= 0 \label{new_added2}.
\end{align}
We further obtain
\begin{align}
&y_0 = x_6 (\bar{x_0} - 2 x_2) \label{dj1} \\
&y_2 = \bar{x_6} (\bar{x_2} -  x_0) \label{dj2}.
\end{align}
 Equation (\ref{dj1}) gives $| \bar{x_0} - 2 x_2 | = 1$, which further implies that $\bar{x_0}  = x_2$.
This is a contradiction to Equation (\ref{dj2}).

\textit{Case A2:} \quad $x_2 = y_2, x_5 =  \bar{y_2} x_6$

We only need to check the validity of Equations (\ref{e2}), (\ref{e5}), and (\ref{e8}). We write them in terms of $5$ variables
$x_0, x_2, x_3, x_6, y_0$:
\begin{align}
x_0 + x_3 \bar{x_2} + \bar{x_2}  - x_3 x_6 \bar{x_0} + y_0 + \bar{x_2} + \bar{x_2} x_3 + x_3 x_6 \bar{y_0} &= 0 \label{1B3} \\
x_0 \bar{x_2} + \bar{x_2} x_6 \bar{x_0} + y_0 \bar{x_2} - \bar{x_3} x_6 + x_3 x_6 - \bar{x_2} x_6 \bar{y_0} &= 0 \nonumber \\
x_0 + x_3 \bar{x_2} - \bar{x_2} +x_3 x_6 \bar{x_0} + y_0 x_3 + \bar{x_2} x_3 - \bar{x_2} - x_6 \bar{y_0} &= 0 \label{1B9}.
\end{align}
In fact, a contradiction can be obtained from Equations (\ref{1B3}) and (\ref{1B9}). Taking the sum and difference of the two equations, we obtain:
\begin{align}
(x_3+1) y_0 + (x_3-1) x_6 \bar{y_0} &= - 4 \bar{x_2} x_3 - 2 x_0  \label{1BAdd} \\
(x_3+1) x_6 \bar{y_0} - (x_3-1) y_0 &= - 4 \bar{x_2} + 2 x_3 x_6 \bar{x_0} \label{1BSub}.
\end{align}

%

If we have a valid solution $ (x_2, x_3, x_0, y_0, x_6)$, it is easy to see that $ (\xi^{-1} x_2, x_3, \xi x_0, \xi y_0, \xi^2 x_6)$ is also a valid solution for any unimodular complex number $\xi$.
Therefore, we only need to consider the case $x_6=1$.
Replacing $x_6=1$ into Equation (\ref{1BSub}), multiplying the equation with $-\bar{x_3}$, and then taking the conjugate, we obtain
\begin{align}
& - (x_3+1) y_0 - (x_3-1) \bar{y_0} =  4 x_2 x_3 - 2 x_0  \label{1BSubChange}.
\end{align}
Adding Equations (\ref{1BAdd}) and (\ref{1BSubChange}) gives
\begin{equation} \label{1BKey}
x_0 = (x_2 - \bar{x_2}) x_3.
\end{equation}
From the identity  $1 = |x_0 | = |x_2 - \bar{x_2}| |x_3| = |x_2 - \bar{x_2}| $, $x_2$ is in the form of
\begin{equation} \label{1Bc2}
\frac{\sqrt{3}}{2} \delta_1 + \frac{i}{2}\delta_2, \quad \delta_1,\delta_2 \in \{ 1, -1 \} .
\end{equation}
Combining Equations (\ref{1Bc2}) and (\ref{1BKey}) gives $x_0 = i \delta_2 x_3 $. Furthermore, Equation (\ref{1BAdd}) is simplified to be
\begin{align}
& (x_3+1) y_0 + (x_3-1) \bar{y_0} = -2 \sqrt{3} \delta_1 x_3 \label{1BContra}.
\end{align}
Equation (\ref{1BContra}) implies that
$$2 \sqrt{3} = |(x_3+1) y_0 + (x_3-1) \bar{y_0}| \leq |x_3+1|+ |x_3-1| \leq 2 \sqrt{2}, $$
which is a contradiction.

\textit{Case A3:} \quad $ x_5 = - x_3 \bar{y_2} x_6, x_2 = x_3 y_2 $

First, we rewrite Equations (\ref{e2}) and (\ref{e8}) in terms of $x_0, x_3, x_6, y_0, y_2$:
\begin{align}
x_0 + 2\bar{y_2} - 2 x_3 \bar{y_2} - x_3 x_6 \bar{x_0} + y_0  +x_3 x_6 \bar{y_0} &=0 \label{1C3} \\
x_0 \bar{x_3} + 2 \bar{x_3} \bar{y_2} + 2 \bar{y_2} + x_6 \bar{x_0} + y_0 - \bar{x_3} x_6 \bar{y_0} &= 0 \label{1C9}.
\end{align}

If we have a valid solution $(y_2, x_3, x_0, y_0, x_6)$, it is easy to see that $ (\xi^{-1} y_2, x_3, \xi x_0, \xi y_0, \xi^2 x_6)$ is also a valid solution for any unimodular complex number $\xi$.
Therefore, we only need to consider the case $x_6=1$.
By computing Equation (\ref{1C3}) + $x_3 \cdot $ Equation (\ref{1C9}) and Equation (\ref{1C3}) - $x_3 \cdot$ Equation (\ref{1C9}) we obtain
\begin{align}
(x_3+1) y_0 + (x_3-1) \bar{y_0} & =  -4 \bar{y_2} - 2 x_0 \label{1CAdd} \\
-(x_3+1) \bar{y_0} + (x_3-1) y_0 &= -4 x_3 \bar{y_2} - 2 x_3 \bar{x_0}  \label{1CSub}.
\end{align}
Multiplying Equation  (\ref{1CSub}) by $\bar{x_3}$, and then taking the conjugate, we obtain
\begin{align}
& - (x_3+1) y_0 - (x_3-1) \bar{y_0} =  -4 y_2 - 2 x_0   \label{1CSubChange}.
\end{align}
Adding Equations  (\ref{1CAdd}) and (\ref{1CSubChange}) we obtain
\begin{equation*}
x_0 = -(y_2 + \bar{y_2}).
\end{equation*}
Thus, $y_2$ is in the form of
\begin{equation*}
y_2 = \frac{1}{2} \delta_1 + \frac{\sqrt{3} i }{2}\delta_2, \quad \delta_1,\delta_2 \in \{ 1, -1 \}.
\end{equation*}
Furthermore,
\begin{align*}
& x_0 = - \delta_1 \\
& (x_3+1) y_0 + (x_3-1) \bar{y_0} = -2 \sqrt{3} \delta_2 ,
\end{align*}
which implies that $2\sqrt{3} \leq 2 \sqrt{2}$.

\textit{Case A4:} \quad  $x_2 = y_2, x_2 = x_3 y_2 $

Similar to case A1, Equations  (\ref{e2}) and (\ref{e8}) imply
\begin{align}
 - (2 \bar{x_2} + x_0) &=  y_0 \label{dj3} \\
x_5 \bar{x_6} - x_6 (\bar{x_0} + x_2) &= 0 \label{dj4}.
\end{align}
Equation  (\ref{dj3}) implies that  $\bar{x_2}  = - x_0$, which is a contradiction to Equation  (\ref{dj4}).

\subsection{Case B:  $x_3 = 1, y_3 = - x_4 \bar{x_6} $}

From Equation  (\ref{e4}) and (\ref{e7}), we obtain
\begin{align}
(\bar{x_5} - y_2 \bar{x_4})  (1-x_2 \bar{y_2}) &=0 \label{dj5} \\
(x_2 \bar{x_6} + y_2 \bar{x_4} ) (1+x_4 \bar{x_5} \bar{y_2}) &=0 \label{dj6}.
\end{align}
If $y_2 \neq x_4 \bar{x_5}$, Equation (\ref{dj5}) gives $x_2 = y_2$. If $x_2 \neq - x_6 \bar{x_4} y_2$, Equation (\ref{dj6}) gives $x_5 = - x_4 \bar{y_2}$. We therefore have the following four cases to discuss.


\textit{Case B1:} \quad  $y_2 = x_4 \bar{x_5}, x_2 = - x_6 \bar{x_4} y_2 = - x_6 \bar{x_5}$

First, we rewrite Equations  (\ref{e5}) and (\ref{e8}) in terms of
$x_0,x_4,x_5,x_6,y_0$:
\begin{align}
2x_6 - x_0 x_5 \bar{x_6} + 2 x_4 + \bar{x_4} x_5 y_0 + x_4 x_5 \bar{x_6} \bar{y_0} +\bar{x_0} x_5 &= 0 \label{2A6} \\
2\bar{x_4} x_5 + x_0 -2 x_5 \bar{x_6} - \bar{x_4}x_6 y_0 + x_4 \bar{y_0} + \bar{x_0} x_6 &=0 \label{2A9}.
\end{align}

%
If we have a valid solution $(x_0, y_0, x_4, x_6, x_5)$, it is easy to see that $ (\xi x_0, \xi y_0, \xi^2 x_4, \xi^2 x_6, \xi^3 x_5)$ is also a valid solution for any unimodular complex number $\xi$.
Therefore, we only need to consider the case $x_6=1$.
Taking  $x_6=1$ into Equations  (\ref{2A6}) and (\ref{2A9}), multiplying Equation (\ref{2A9}) by $-\bar{x_5}$, and taking its conjugate, we obtain
\begin{align}
2 - x_0 x_5 + 2 x_4 + \bar{x_4} x_5 y_0 + x_4 x_5 \bar{y_0} + \bar{x_0} x_5 &= 0 \label{2A6Change} \\
2 - x_0 x_5 -2 x_4 - \bar{x_4} x_5 y_0 + x_4 x_5 \bar{y_0} - \bar{x_0} x_5 &= 0. \label{2A9Change}
\end{align}
Adding Equations (\ref{2A6Change}) and (\ref{2A9Change}) gives
\begin{equation} \label{new_add2}
4 = 2 x_0 x_5 - 2 x_4 x_5 \bar{y_0}.
\end{equation}
Because $ |2x_0 x_5 | = | 2 x_4 x_5 \bar{y_0} | = 2 $, the only possibility is
\begin{equation} \label{2AContra}
x_0 x_5 =1, x_4 x_5 \bar{y_0} = -1.
\end{equation}
Taking Equation  (\ref{2AContra}) into (\ref{2A6Change}) we obtain $2x_4 = 0$, which is a contradiction.

\textit{Case B2:} \quad  $x_2 = y_2, x_2 = - x_6 \bar{x_4} y_2$

Clearly,  $x_6 = -x_4$.
Because $y_3 = - x_4 \bar{x_6} = 1$, this case is covered by Case A.

\textit{Case B3:} \quad $y_2 = x_4 \bar{x_5}, x_5 = - x_4 \bar{y_2}$

Clearly, $y_2 = -x_4 \bar{x_5}$.
Because $y_2 = x_4 \bar{x_5} = -y_2 $,This case is not possible.

\textit{Case B4:} \quad $x_2 = y_2, x_5 = - x_4 \bar{y_2}= - x_4 \bar{x_2} $

First, we rewrite Equations  (\ref{e2}) and (\ref{e5}) in terms of
$x_0,x_2,x_4,x_6,y_0$:
\begin{align}
2 \bar{x_2} - 2 \bar{x_2} x_4 \bar{x_6} + x_0 + y_0 + x_4 \bar{x_0} - x_4\bar{y_0} &= 0 \label{2D3} \\
2 x_4 + 2x_6 + x_0 \bar{x_2} + \bar{x_2} y_0 - \bar{x_2} x_4 \bar{x_0} + \bar{x_2} x_4 \bar{y_0}  &= 0\label{2D6}.
\end{align}
%
%
If we have a valid solution $(x_2, x_0, y_0, x_4, x_6)$, it is easy to see that $ (\xi^{-1} x_2, \xi x_0, \xi y_0, \xi^2 x_4, \xi^2 x_6)$ is also a valid solution for any unimodular complex number $\xi$.
Therefore, we only need to consider the case $x_4=1$.
Taking the conjugate of Equation  (\ref{2D3}), multiplying Equation  (\ref{2D6}) by $\bar{x_2}$, and letting $x_4 = 1$, we have
\begin{align}
& 2x_2 - 2 x_2 x_6 + \bar{x_0} + \bar{y_0} + x_0 - y_0 = 0 \label{2D3Change} \\
& 2x_2 + 2 x_2 x_6 +x_0 + y_0 - \bar{x_0} +\bar{y_0} = 0. \label{2D6Change}
\end{align}
Adding Equations  (\ref{2D3Change}) and (\ref{2D6Change}), we obtain
\begin{equation} \label{new_add3}
4x_2 + 2 x_0 + 2\bar{y_0} = 0.
\end{equation}
Because $|4x_2|=4, |2x_0| = | 2 \bar{y_0} | = 2 $, the only possibility is
\begin{equation} \label{2DContra}
x_0 = -x_2,  \bar{y_0} = -x_2.
\end{equation}
Applying Equation  (\ref{2DContra}) to (\ref{2D3Change}) gives $2x_2 x_6 = 0$, which is a contradiction.

\subsection{Case C: $x_6 = x_4, y_3 = - x_3$ }

From Equations  (\ref{e4}) and (\ref{e7}), we obtain
\begin{align}
 (\bar{x_5} - \bar{x_4} y_2) (1-x_2 \bar{y_2}) &=0  \label{dj7} \\
 (\bar{x_5} x_3 + \bar{x_4} y_2 )(x_2 + y_2) &=0   \label{dj8}.
\end{align}
If $y_2 \neq x_2$, Equation (\ref{dj7}) gives  $y_2 = x_4 \bar{x_5}$. If $y_2 \neq -x_2$, Equation (\ref{dj8}) gives $y_2 = -x_3 x_4 \bar{x_5} $.
So there are four cases to consider.

\textit{Case C1:} \quad $y_2 = x_2, y_2 = -x_3 x_4 \bar{x_5}$

First, we rewrite Equations  (\ref{e5}) and (\ref{e8}) in terms of
$x_0,x_2,x_3,x_4,y_0$:
\begin{align}
2x_4+2 \bar{x_3} x_4 + \bar{x_2}  x_0 + \bar{x_2} y_0 - \bar{x_2} x_3 x_4 \bar{x_0} + \bar{x_2} x_3 x_4\bar{y_0} &= 0 \label{3A6} \\
2\bar{x_2} - 2 \bar{x_2} x_3+ \bar{x_c} x_0 - \bar{x_3} y_0+x_4\bar{x_0} + x_4 \bar{y_0}  &= 0 \label{3A9}.
\end{align}
If we have a valid solution $(x_2, x_3, x_0, y_0, x_4)$, it is easy to see that $ (\xi^{-1}  x_2, x_3, \xi  x_0, \xi y_0, \xi^2 x_4)$ is also a valid solution for any unimodular complex number $\xi$.
Therefore, we only need to consider the case $x_4=1$.
Taking the conjugate of Equation  (\ref{3A9}), multiplying it by $\bar{x_2}$, and  letting $x_4 = 1$, Equations  (\ref{3A6})
and (\ref{3A9}) give
\begin{align}
2+ 2 \bar{x_3} + \bar{x_2}x_0 + \bar{x_2}y_0 - \bar{x_2}x_3\bar{x_0} +\bar{x_2}x_3\bar{y_0} &= 0 \label{3A6Change} \\
2- 2 \bar{x_3} + \bar{x_2}x_0 + \bar{x_2}y_0 +\bar{x_2}x_3\bar{x_0} -\bar{x_2}x_3\bar{y_0} &= 0. \label{3A9Change}
\end{align}
Adding Equations (\ref{3A6Change}) and (\ref{3A9Change}) gives
\begin{equation} \label{new_add4}
4 + 2 \bar{x_2} (x_0+y_0) = 0.
\end{equation}
Thus $ |x_0+y_0| = 2 $, the only possibility is
\begin{equation} \label{3AContra}
x_0 =  y_0 = -x_2.
\end{equation}
Applying Equation  (\ref{3AContra}) to (\ref{3A6Change}), we obtain $2\bar{x_3}= 0$, which is a contradiction.

\textit{Case C2:} \quad $y_2 = x_4 \bar{x_5}, y_2 = -x_2$

First, we rewrite Equations  (\ref{e2}) and (\ref{e5}) in terms of
$x_0,x_2,x_3,x_4,y_0$:
\begin{align}
2x_3 \bar{x_2} -2\bar{x_2}  +x_0 + y_0 + x_4 \bar{x_0} - x_4\bar{y_0} &= 0 \label{3B3} \\
2x_4 +2x_4 \bar{x_3}+x_0\bar{x_2} -y_0 \bar{x_2} -\bar{x_2}x_4\bar{x_0} - \bar{x_2}x_4 \bar{y_0}  &= 0 \label{3B6}.
\end{align}
If we have a valid solution $(x_2, x_3, x_0, y_0, x_4)$, it is easy to see that $ (\xi^{-1}  x_2, x_3, \xi  x_0, \xi y_0, \xi^2 x_4)$ is also a valid solution for any unimodular complex number $\xi$.
Therefore, we only need to consider the case $x_4=1$.
Taking the conjugate of Equation  (\ref{3B6}), multiplying it by $\bar{x_2}$, and letting $x_4 = 1$, Equations  (\ref{3B3}) and (\ref{3B6}) give
\begin{align}
 2x_3 \bar{x_2} -2\bar{x_2}  +x_0 + y_0 + \bar{x_0} - \bar{y_0}&= 0 \label{3B3Change} \\
 2x_3 \bar{x_2} +2\bar{x_2}  - x_0 - y_0 +\bar{x_0} - \bar{y_0}&= 0. \label{3B6Change}
\end{align}
Subtracting Equations (\ref{3B6Change}) and (\ref{3B3Change}), we obtain
\begin{equation} \label{new_add5}
4  \bar{x_2} - 2 (x_0+y_0) = 0.
\end{equation}
Thus $ |x_0+y_0| = 2 $, and the only possibility is
\begin{equation} \label{3BContra}
x_0 =  y_0 = \bar{x_2}.
\end{equation}
Applying Equation  (\ref{3BContra}) to (\ref{3B3Change}), we obtain $2x_3 \bar{x_2}= 0$, which is a contradiction.

\textit{Case C3:} \quad $y_2 = x_4 \bar{x_5}, y_2 = -x_3 x_4 \bar{x_5} $
This case is covered by Case A, because $x_3 = -1, y_3 = -x_3 = 1$.

\textit{Case C4:} \quad $y_2 = x_2, y_2 = -x_2$
This case is clearly not possible.

\section{Proof of Theorem \ref{M26PointNonexistenceThm}} \label{appendix_a2}

The proof follows a similar procedure to that of the proof of Theorem \ref{M2NonexistenceThm}.
The only difference is that the modulus of $x_0,y_0$  are changed from one to zero.
Related changes in the proof in Appendix \ref{appendix_a} are listed below.

Case A1: Equation (\ref{new_added1}) gives $2 \bar{x_2}  = 0$, which is a contradiction.

Case A2: Equation (\ref{1BAdd})  gives $4 \bar{x_2} x_3 = 0$, which is a contradiction.

Case A3: Equation (\ref{1CAdd})  gives $4 \bar{y_2} = 0$, which is a contradiction.

Case A4: Equation (\ref{dj3})  gives $ 2 \bar{x_2} = 0$, which is a contradiction.

Case B1: Equation (\ref{new_add2})  gives $4 = 0$, which is a contradiction.

Case B4: Equation (\ref{new_add3}) gives $4x_2 = 0$, which is a contradiction.

Case C1: Equation (\ref{new_add4}) gives $4 = 0$, which is a contradiction.

Case C2: Equation (\ref{new_add5}) gives $4  \bar{x_2} = 0$, which is a contradiction.

\section{Proof of Theorem \ref{complexConjecture}} \label{appendix_a3}

Let $\mathbb{Z}[\lambda]$ denote the polynomial ring over $ \mathbb{Z}$, and $\Phi_N(\lambda)$ denote the $N$th cyclotomic polynomial. 
Let $\xi = \exp(i 2 \pi / N)$, $U_N=\{\xi^j \mid j=0,\cdots,N-1\}$ be the group of $N$th roots of unity endowed with multiplication. 
For any $\eta \in U_N$, let $|\eta |$ denote the order of $\eta$ in the cyclic group $U_N$. 
Because $\sum_{m=1}^M x_m = 0$, there exists a polynomial $F(\lambda) = \sum_{j=0}^{N-1} f_j \lambda^j \in \mathbb{Z}[\lambda]$ such that $f_j \geq 0$ and $F(\xi) = 0$.  

We prove Theorem \ref{complexConjecture} using a sequence of lemmas.

\begin{lemma} \label{lemma_seq0}
Let  $p_k$ be distinct prime numbers and integers $r_k > 0, k = 1,2$. Then
\begin{align}
\Phi_{p_1}(\lambda^{p_1^{r_1-1}p_2^{r_2}} )  &= \prod\limits_{i=0}^{r_2}   \Phi_{p_1^{r_1} p_{2}^{i} }(\lambda ),   \label{cyclo_lemma2} \\
\Phi_{p_1}(\lambda^{p_1^{r_1-1}p_2^{r_2} })  &= \Phi_{p_1}(\lambda^{p_1^{r_1-1}p_2^{r_2-1}} )  \Phi_{p_1^{r_1}p_2^{r_2}}(\lambda), \label{time2} \\
 \Phi_{p_1}(\lambda^{p_1^{r_1-1}} ) &= \Phi_{p_1^{r_1}}(\lambda) . \label{time20}
\end{align}
Similar results hold if $p_1$ and $r_1$ are respectively  replaced with $p_2$ and $r_2$ in the above equations.  
\end{lemma}

\textit{Proof of Lemma \ref{lemma_seq0}}:

Since both sides are monic and have degree $(p_1-1)p_1^{r_1-1} = (p_1-1)p_1^{r_1-1}p_2^{r_2} $, 
it suffices to show that every zero of
$\Phi_{p_1^{r_1} p_2^{i_2}}(\lambda )$  is a zero of $\Phi_{p_1}(\lambda^{p_1^{r_1-1}} ) $.
If $\eta$ is a zero of $\Phi_{p_1^{r_1} p_2^{i_2} }(\lambda )$, then $|\eta| = p_1^{r_1} p_2^{i_2} $, which implies $|\eta^{p_1^{r_1-1}}| = |\eta|  / \textrm{gcd}(|\eta|,p_1^{r_1-1}) = p_1 $. Therefore, $\eta$ is also a zero of $\Phi_{p_1}(\lambda^{p_1^{r_1-1}} ) $.

The proof of Equations (\ref{time2}) and (\ref{time20}) is similar. 

\begin{lemma} \label{lemma_seq1}
If $n=1$, i.e. $N=p_1^{r_1}$, 
then there exists a polynomial  $A(\lambda) = \sum_{j=0}^{N/p_1-1}  a_{j} \lambda^j \in \mathbb{Z}[\lambda], \, a_{j} \geq 0$, such that 
$$
F(\lambda ) = \Phi_{p_1}(\lambda^{\frac{N}{p_1}} )  A(\lambda).
$$
\end{lemma}

\textit{Proof of Lemma \ref{lemma_seq1}}:

First, $\Phi_N(\lambda)$ divides $F(\lambda)$, because $F(\lambda)$ annihilates $\xi$ and $\Phi_N(\lambda)$ is an irreducible and monic polynomial in the ring $\mathbb{Z}[\lambda]$. 
Besides this, Equation (\ref{time20}) gives $\Phi_{p_1}(\lambda^{N/p_1} ) = \Phi_N(\lambda)$.
Therefore, there exists a polynomial $A_k(\lambda) =\sum_{j=0}^{N/ p_1-1}  a_{j} \lambda^j  \in \mathbb{Z}[\lambda]$ such that $F(\lambda ) = \Phi_{p_1}(\lambda^{N/p_1} )  A(\lambda).$

Second, because $\deg (F) < N$, we have $\deg (A) = \deg (F) - \deg (\Phi_N(\lambda) )< p_1^{r_1-1}=N/p_1$. We note that $f_{j} = a_j, j = 0, \cdots, N/ p_1-1$ and that $f_j\geq 0$. Therefore, the coefficients of $A(\lambda)$ are nonnegative.

\begin{lemma} \label{lemma_seq2}

If $n=2$, i.e. $N = p_1^{r_1}  p_2^{r_2}$,  
then there exist polynomials $\hat{A}_k(\lambda) \in \mathbb{Z}[\lambda], k=1,2$ such that 
 \begin{align} \label{time1}
 \hat{A}_1(\lambda) \Phi_{p_1}(\lambda^{\frac{N}{p_1}} ) + \hat{A}_2(\lambda) \Phi_{p_2}(\lambda^{\frac{N}{p_2}} ) =\Phi_N(\lambda).
 \end{align}

\end{lemma}

 \textit{Proof of Lemma \ref{lemma_seq2}}:

 First, Equation (\ref{cyclo_lemma2}) and its similar result (by replacing $p_1, r_1$ with $p_2,r_2$) imply that $\textrm{gcd}( \Phi_{p_1}(\lambda^{N/p_1} ), \Phi_{p_2}(\lambda^{N/p_2} ) ) =\Phi_N(\lambda)$. 
 
 Second, consider two polynomials $T_{t_k}(\lambda)=1+\lambda+\cdots + \lambda^{t_k-1}, \, k=1,2$, where $t_1 > t_2$ and $\textrm{gcd}(t_1,t_2)=1$. We apply Euclidean division to $t_1, t_2$ to obtain $t_1 = t_2 q+ b, \, 0 < b < t_2$. 
It is easy to observe that  $T_{t_1}(\lambda) = T_{t_2}(\lambda)\sum_{j=1}^q \lambda^{t_1-j t_2} + T_{b}(\lambda)$. If we continuously apply Euclidean division, we will find  polynomials $\hat{A}_k(\lambda) \in \mathbb{Z}[\lambda], k=1,2$ such that 
\begin{align} \label{time3}
\hat{A}_1(\lambda) T_{t_1}(\lambda) + \hat{A}_2(\lambda) T_{t_2}(\lambda)= 1
 \end{align}
  Replacing  $t_1,t_2,$ and $\lambda$ respectively by $p_1,p_2,$ and $\lambda^{N/(p_1 p_2)}$ in Equation (\ref{time3}), multiplying both sides by $\Phi_N(\lambda)$, and using Equation (\ref{time2}) and its similar result, we obtain Equation (\ref{time1}).

\begin{lemma} \label{lemma_seq3}

If $n=2$, i.e. $N = p_1^{r_1}  p_2^{r_2}$, then there exist polynomials $A_k(\lambda), \deg (A_k) \leq N/p_k - 1, k = 1, 2$ such that $F(\lambda) = \sum_{k=1}^{2} A_k(\lambda) H_k(\lambda)$, where  $H_k(\lambda) = \Phi_{p_k}(\lambda^{N/p_k} ) ,k = 1, 2$.

\end{lemma}

 \textit{Proof of Lemma \ref{lemma_seq3}}:
 
Clearly, $\Phi_N(\lambda)$ divides $F(\lambda)$ due to the reason  mentioned before. 
Therefore, Lemma \ref{lemma_seq2} implies that there exist polynomials $\widehat{A_k}(\lambda) \in \mathbb{Z}[\lambda], k = 1, 2$ such that $F(\lambda) = \sum_{k=1}^{2} \widehat{A_k}(\lambda) H_k(\lambda)$ holds.

It is easy to see that $ \lambda^d H_k(\lambda)$ can be written as $ \lambda^d H_k(\lambda) = (\lambda^N-1) Q(\lambda ) +\lambda^{d_0} H_k(\lambda)  $ for some $Q(\lambda ) \in \mathbb{Z}[\lambda], 0 \leq d_0 \leq N/p_k - 1$. Thus, there exist polynomials $A_k(\lambda), \deg (A_k) \leq N/p_k - 1 , k = 1, 2$, and $W(\lambda )$ such that 
$$F(\lambda) = \sum_{k=1}^{2} A_k(\lambda) H_k(\lambda) +  W(\lambda )  (\lambda^{N}-1 ).$$  Since $\deg (F) \leq N/p_k - 1 $, we obtain $W(\lambda ) = 0$. 
 
\begin{lemma} \label{lemma_seq4}

 The coefficients of $A_k(\lambda), k=1,2$ in Lemma \ref{lemma_seq3} can be made nonnegative. 

\end{lemma}

 \textit{Proof of Lemma \ref{lemma_seq4}}:

Let 
\begin{align*}
\mathfrak{D}_F &= \left \{[ a_{10}, \cdots a_{1(N/p_1 - 1)} ,   a_{2 0} , \cdots , a_{2 (N/p_2 - 1)}]   \mid 
F(\lambda) =  \sum_{j_1=0}^{\frac{N}{p_1}-1}  a_{1j_1}  H_1(\lambda) + \sum_{j_2=0}^{\frac{N}{p_2}-1}  a_{2j_2}  H_2(\lambda).
\right \}  \\
\mathfrak{D}_F^{(2)} &= \left \{ [ a_{10}, \cdots a_{1(N/p_1 - 1)} ,   a_{2 0} , \cdots , a_{2 (N/p_2 - 1)}]   \mid 
F(\lambda) =  \sum_{j_1=0}^{\frac{N}{p_1}-1}  a_{1j_1}  H_1(\lambda) + \sum_{j_2=0}^{\frac{N}{p_2}-1}  a_{2j_2}  H_2(\lambda), \, a_{2j_2} \geq 0, j_2=0, \cdots, \frac{N}{p_2}-1.
\right \}
\end{align*}
For a fixed integer $0 \leq j \leq N/p_1-1$, consider the set $\{j+k N/p_1, k = 0, \cdots, p_1-1\} \, (\textrm{mod }N)$. 
For each $k = 0, \cdots, p_1-1 $,  we apply Euclidean division to $j+k N/p_1$ and $N/p_2$, and obtain
 integers $ 0 \leq g_k \leq N/p_2-1, 0 \leq h_k \leq p_2-1$ such that $j+k N/p_1 = g_k + h_k N/p_2$. Because of the identity
$$\bigcup_{\tau = 0}^{p_2-1}  \{g_k + (h_k + \tau) N/p_2, k = 0, \cdots, p_1-1\}
= \bigcup_{k = 0}^{p_1-1} \{g_k + (h_k + \tau) N/p_2, \tau = 0, \cdots, p_2-1\}  \quad (\textrm{mod }N),$$
within the set $\mathfrak{D}_F$ 
we can always  decrease $  a_{1 (j + \tau N/p_2 \, (\textrm{mod }N/p_1))}, \tau = 0,\cdots, p_2-1$ by one, while increasing $a_{2g_k }, k=0,\cdots, p_1-1$ by one. 
We conclude that the subset  $\mathfrak{D}_F^{(2)}$ of $\mathfrak{D}_F$ is not empty. 

To finish the proof, it suffices to show that within $\mathfrak{D}_F^{(2)}$, there exist an element with $a_{1j_1} \geq 0$ for all  $j_1=0,\cdots,  N/p_1 - 1 $. 
If this is not true, then there exists $\mu < 0$ such that 
\begin{align} \label{time30}
\mu = \max_{ [ a_{10}, \cdots a_{1(N/p_1 - 1)} ,   a_{2 0} , \cdots , a_{2 (N/p_2 - 1)}] \in \mathfrak{D}_F^{(2)} } 
	\, \bigg\{ \min \left \{  a_{10}, \cdots , a_{1(N/p_1 - 1)}  \right \}  \bigg\}< 0.
\end{align}
Suppose that  $a_{1j} = \mu$. For each $k=0,\cdots, p_1-1$, consider the nonnegative coefficient of the item $\lambda^{j + kN/p_1}$ in $F(\lambda)$:  $A_1(\lambda)H_1(\lambda)$ contributes  a negative value $  a_{1 j}$ to it, and thus $A_2(\lambda)H_2(\lambda)$ contributes a positive value. 
In other words, there exist integers $ 0 \leq g_k \leq N/p_2-1, 0 \leq h_k \leq p_2-1$ such that $j+k N/p_1 = g_k + h_k N/p_2$ and that $a_{2 g_k} > 0$.
It is clear that $g_k, k=0 , \cdots, p_1-1$ are distinct values. By similar reasoning as before, we can  increase $  a_{1 (j + \tau N/p_2 \, (\textrm{mod }N/p_1))}, \tau = 0,\cdots, p_2-1$ by one, while decreasing $a_{2g_k}, k=0,\cdots, p_1-1$ by one, in order to get another  element in $\mathfrak{D}_F^{(2)}$. 
Thus, we can increase 
$\max_{ [ a_{10}, \cdots a_{1(N/p_1 - 1)} ,   a_{2 0} , \cdots , a_{2 (N/p_2 - 1)}] \in \mathfrak{D}_F^{(2)} } 
	\, \bigg\{ \min \left \{  a_{10}, \cdots , a_{1(N/p_1 - 1)}  \right \}  \bigg\}$,  contradicting the definition of $\mu$ in (\ref{time30}). 

\vspace{0.3cm}

 \textit{Proof of Theorem \ref{complexConjecture} }:

Combining Lemmas \ref{lemma_seq1} to \ref{lemma_seq4}, we conclude that $F(\lambda)$ can be written as:
\begin{align*} 
F(\lambda) &= \sum_{k=1}^{2} A_k(\lambda) H_k(\lambda) , \textrm{ where } A_k(\lambda) = \sum_{j=0}^{\frac{N}{p_k}-1}  a_{kj} \lambda^j \in \mathbb{Z}[\lambda], \, a_{kj} \geq 0, \\
H_k(\lambda) &= \Phi_{p_k}(\lambda^{N/p_k} ) = 1 + \lambda^{\frac{N}{p_k}} +  \lambda^{2\frac{N}{p_k}} + \cdots + \lambda^{ (p_k-1) \frac{N}{p_k}}, 
\end{align*}
which is equivalent to Equation (\ref{time10}). Equation  (\ref{Thm3_2}) then immediately follows from Equation (\ref{time10}). 

\section*{Acknowledgment}

The authors thank Dr. Yaron Rachlin of the MIT Lincoln Laboratory for valuable comments and discussions.

\ifCLASSOPTIONcaptionsoff
  \newpage
\fi

\bibliographystyle{IEEEtran}
\bibliographystyle{unsrt}
\bibliography{hex,new_reference,CDS}

\begin{IEEEbiographynophoto}{Jie Ding}
is a Ph.D. candidate in the School of Engineering and Applied Sciences, Harvard University.
His current research are in cyclic difference sets, time series, information theory.
\end{IEEEbiographynophoto}


\begin{IEEEbiographynophoto}{Mohammad Noshad}
%
is a postdoctoral fellow in the Electrical Engineering Department at Harvard University. He received his PhD in Electrical Engineering from the University of Virginia. He is a recipient of the ``Best Paper Award'' at IEEE Globecom 2012. His research interests include information and coding theory, statistical machine learning, free-space optical communications, and visible light communications.
\end{IEEEbiographynophoto}

\begin{IEEEbiographynophoto}{Vahid Tarokh}
 is a professor of applied mathematics in the School of Engineering and Applied Sciences, Harvard University.
His current research interests are in data analysis, network security, optical surveillance, and radar theory.
\end{IEEEbiographynophoto}




\end{document}